\documentclass[11pt,a4paper]{article}
\usepackage{amsmath,amsfonts,amsthm,times,latexsym,graphicx,datetime}
\usepackage{natbib}

\newtheoremstyle{localthm}
	{5pt} 
	{5pt} 
	{\sl} 
	{} 
	{\bf} 
	{} 
	{.7em} 
	{} 

\newtheoremstyle{localrem}
	{5pt} 
	{5pt} 
	{\rm} 
	{} 
	{\bf} 
	{} 
	{.7em} 
	{} 

\theoremstyle{localthm}
\newtheorem{Theorem}{Theorem}[section]
\newtheorem{Lemma}[Theorem]{Lemma}
\newtheorem{Corollary}[Theorem]{Corollary}

\theoremstyle{localrem}

\DeclareMathOperator*{\argmax}{arg\,max\hspace*{2pt}}

\newcommand{\dom}{\mathrm{dom}}

\def\bs{\boldsymbol}

\def\R{\mathbb{R}}

\def\TT{\mathcal{T}}

\def\eps{\varepsilon}

\def\til{\widetilde}

\def\Ex{\mathop{\rm I\!E}\nolimits}
\def\Pr{\mathop{\rm I\!P}\nolimits}
\def\Var{\mathrm{Var}}

\begin{document}

\addtolength{\baselineskip}{+.48\baselineskip}

\title{Maximum-Likelihood Estimation of a Log-Concave Density based on Censored Data$^*$}
\author{Lutz D\"{u}mbgen (University of Bern)\\
	Kaspar Rufibach (Roche Biostatistics Oncology, Basel)\\
	Dominic Schuhmacher (University of G\"ottingen)}

\date{\today}

\maketitle

\begin{abstract}
We consider nonparametric maximum-likelihood estimation of a log-concave density in case of interval-censored, right-censored and binned data. We allow for the possibility of a subprobability density with an additional mass at $+\infty$, which is estimated simultaneously. The existence of the estimator is proved under mild conditions and various theoretical aspects are given, such as certain shape and consistency properties. An EM algorithm is proposed for the approximate computation of the estimator and its performance is illustrated in two examples.
\end{abstract}

\bigskip

\noindent
$^*$ 
Work supported by research group FOR916 of Swiss National Science Foundation and Deutsche Forschungsgemeinschaft.



\section{Introduction}
\label{sec:intro}

We consider estimation of an unknown distribution $P$ on $(-\infty,\infty]$ based on data which are ``censored'' in a rather general sense. We assume that $q := P(\{\infty\})$ is a number in $[0,1)$ and that $P$ has a log-concave sub-probability density $f$ on $\R$. This means that $f = e^\phi$ for some concave function $\phi : \R \to [-\infty,\infty)$ with $\int e^{\phi(x)} \, dx = 1 - q$, and
\[
	P(B) = P_{\phi,q}(B) \ = \ \int_B e^{\phi(x)} \, dx
			+ 1_{[\infty \in B]} \, q
\]
for any Borel set $B \subset (-\infty,\infty]$.

In the simplest setting our data consists of independent observations $X_1,X_2,\ldots,X_n$ drawn from $P$. For $q = 0$ this case was investigated in detail in \cite{Duembgen_Rufibach_2009}. As explained in the latter paper, the shape constraint of log-concavity is rather natural in many situations and leads to enhanced estimators of the distribution function of $P$ as well as good estimators of the density $f$ without requiring the choice of any tuning parameter. See also the review of \cite{Walther_2009} about the benefits and possible applications of log-concavity.

In many applications the values $X_i$ are not exactly observed. One well-known example is right-censoring: Suppose that the $X_i$ are event times in a biomedical study with values in $(0,\infty]$, i.e.\ $P((-\infty,0]) = 0$ and $\phi(x) = -\infty$ for $x < 0$. Here $X_i = \infty$ means that the event does not happen at all, and $q$ is sometimes referred to as the ``cure parameter''. If the study ends at time $C_i$ from the viewpoint of the $i$-th unit but $X_i > C_i$, then we have a right-censored observation and know only that $X_i$ is contained in the interval $\til{X}_i = (C_i, \infty]$. In other settings one has purely interval-censored data: The $i$-th unit is inspected at one or several time points, and at each inspection one can only tell whether the event in question has already happened or not. This gives also an interval $\til{X}_i = (L_i, R_i] \subset (0,\infty]$ containing $X_i$. Related to interval-censoring is rounding or binning: For a given partition of $(-\infty,\infty]$ into nondegenerate left-open and right-closed intervals, we only know which interval observation $X_i$ belongs to. In view of econometric applications (e.g.\ log-returns, log-incomes) it is desirable to allow negative values of the $X_i$. Whenever we talk about ``censored data'' we mean right-censored, interval-censored, binned or rounded data. The censoring or inspection time points or the binning intervals are assumed to be either fixed or random and independent from the random variables $X_i$.

In case of censored data, the potential benefits of shape-constraints are even higher than in settings with complete data. To analyze interval-censored data, \cite{Duembgen_Freitag_Jongbloed_2006} constrained the density $f$ on $[0,\infty)$ to be non-increasing or unimodal. The former constraint leads typically to accelerated rates of convergence compared to the unrestricted nonparametric estimator, see for instance \cite{Duembgen_Freitag_Jongbloed_2004}. An obvious question is how we can cope with the constraint of $f$ being log-concave, which is stronger than $f$ being unimodal.

The remainder of this paper is organized as follows: In Section~\ref{sec:likelihood} we introduce the log-likelihood functions for our general setting and provide necessary and sufficient conditions for the existence of maximizers. In Section~\ref{sec:parameters} we show how the parameter space may be restricted and approximated. Particular algorithms for the computation of the MLEs are proposed in Section~\ref{sec:algorithms}. They utilize the EM paradigm of \cite{Dempster_etal_1977} and the fast algorithms for complete data by \cite{Duembgen_etal_2007}. Section~\ref{sec:identifiability.consistency} discusses (partial) identifiability of the special parameter $q$ and some consistency properties of our estimators. In Section~\ref{sec:examples} we illustrate our methods with real and simulated data. Proofs and technical details are deferred to Section~\ref{sec:proofs}.

\section{Log-likelihoods and maximum-likelihood estimators}
\label{sec:likelihood}

\paragraph{Log-likelihood functions.}
Our full parameter space $\Theta$ is the set of all pairs $(\phi,q)$ consisting of a concave and upper semicontinuous function $\phi : \R \to [-\infty,\infty)$ and a parameter $q \in [0,1)$ such that
\begin{equation}
\label{eq:constraint Theta}
	\int e^{\phi(x)} \, dx + q \ = \ 1 .
\end{equation}
If we fix the value $q$, the set of all concave and upper semicontinuous functions $\phi$ satisfying \eqref{eq:constraint Theta} is denoted by $\Phi(q)$.

If we could observe the random variables $X_1, X_2, \ldots, X_n$, an appropriate normalized log-likelihood function $\tilde{\ell} : \Theta \to [-\infty,\infty)$ would be given by
\begin{equation}
\label{eq: log-lik-complete-data}
	\tilde{\ell}(\phi, q) \ := \ \frac{1}{n}
		\sum_{i=1}^n \bigl( 1_{[X_i < \infty]} \phi(X_i)
			+ 1_{[X_i = \infty]} \log q \bigr) .
\end{equation}
In case of censored data we observe random subintervals $\til{X}_1$, $\til{X}_2$, \ldots, $\til{X}_n$ of $(-\infty,\infty]$. More precisely, we assume that either $\til{X}_i = (L_i, R_i] \ni X_i$ with $-\infty < L_i < R_i \le \infty$, or $\til{X}_i$ consists only of the one point $L_i = R_i = X_i \in \R$. Note that we exclude the possibility of $L_i = -\infty$, which is convenient and typically no serious restriction. For instance, in connection with event times $X_i > 0$, the left end points $L_i$ are always nonnegative.

After conditioning on all censoring and inspection time points or binning intervals, we end up with independent observations $\til{X}_i$, and the normalized log-likelihood function $\ell : \Theta \to [-\infty,\infty)$ for our setting is given by
\begin{equation}
\label{eq: log-lik-cens-data}
	\ell(\phi, q) \ := \ \frac{1}{n}
		\sum_{i=1}^n \bigl( 1_{[L_i = R_i]} \phi(X_i)
			+ 1_{[L_i < R_i]} \log P_{\phi,q}((L_i,R_i]) \bigr) .
\end{equation}

Sometimes we want to rule out the possibility of a positive mass $q$ at infinity, in which case we consider
\[
	\ell(\phi) \ := \ \ell(\phi,0)
\]
for $\phi \in \Phi(0)$.

\paragraph{Maximum-likelihood estimators.}
Our goal is to find a maximum-likelihood estimator (MLE) $(\hat{\phi},\hat{q})$ of $(\phi,q)$, i.e.\ a maximizer of $\ell(\cdot,\cdot)$ over $\Theta$. Under the restriction that $q = 0$ we aim to find a MLE $\hat{\phi}_0$ of $\phi$, i.e.\ a maximizer of $\ell(\cdot)$ over $\Phi(0)$.

Our first theorem characterizes the existence of these MLEs.

\begin{Theorem}[Existence of MLEs]
\label{thm:existence}
A maximizer $\hat{\phi}_0$ of $\ell(\cdot,0)$ over $\Phi(0)$ exists if, and only if, there exists no uncensored observation $\tilde{X}_{i_o} = \{X_{i_o}\}$ such that each interval $[L_i,R_i]$ contains $X_{i_o}$.

A maximizer $(\hat{\phi},\hat{q})$ of $\ell(\cdot)$ over $\Theta$ exists if, and only if, there exists no uncensored observation $\tilde{X}_{i_o} = \{X_{i_o}\}$ such that each interval $[L_i,R_i]$ contains $X_{i_o}$ or $\infty$.
\end{Theorem}

Note that the MLEs may only fail to exist in situations where the exact observations $\{L_i \colon L_i=R_i, 1\leq i \leq n\}$ form a one-point set. Therefore both MLEs $\hat{\phi}_0$ and $(\hat{\phi},\hat{q})$ exist in the case of purely interval-censored, rounded or binned data. In the classical right-censored case, assuming i.i.d.\ censoring times $C_1, \ldots, C_n$ and writing $\varrho := \Pr(X_i \leq C_i)$, the probability for existence of both MLEs is at least $1-n\varrho(1-\varrho)^{n-1}$, which goes to $1$ geometrically fast.

In the first part of Section~\ref{sec:parameters} we describe some simple special cases in which the MLE $(\hat{\phi}, \hat{q})$ either does not exist or is rather trivial.

\section{Restricting and approximating the parameter spaces}
\label{sec:parameters}

\paragraph{Special cases.}
In some situations a MLE $(\hat{\phi},\hat{q})$ may not exist or may be rather trivial. The next two lemmas describe such scenarios.

\begin{Lemma}
\label{lem:exotic1}
Suppose that
\[
	\bigcap_{i=1}^n [L_i,R_i] \ = \ [\mu',\mu'']
\]
for certain numbers $-\infty < \mu' < \mu'' \le \infty$. Then $\ell(\phi,q) \le 0$ with equality if, and only if, $P_{\phi,q}((\mu',\mu'']) = 1$.
\end{Lemma}

\begin{Lemma}
\label{lem:exotic2}
Suppose that
\[
	\bigcap_{i=1}^n [L_i,R_i] \ = \ \{\mu\}
\]
for some point $\mu \in \R$. If $L_i = R_i = \mu$ for at least one index $i$, then
\[
	\sup_{\phi \in \Phi(0)} \ell(\phi) \ = \ \infty .
\]
Otherwise, let $n_\ell := \# \{i : L_i < \mu = R_i\}$, $n_r := \# \{i : L_i = \mu < R_i\}$, and define $a := \max\{L_i : L_i < \mu\}$, $b := \min\{R_i : R_i > \mu\}$. Then
\[
	\ell(\phi,q)
		\ \le \ \frac{n_\ell}{n} \log \Bigl( \frac{n_\ell}{n_\ell + n_r} \Bigr)
			+ \frac{n_r}{n} \log \Bigl( \frac{n_r}{n_\ell + n_r} \Bigr)
\]
with equality if, and only if,
\begin{equation}
\label{eq:just two intervals}
	P_{\phi,q}((a,\mu]) \ = \ \frac{n_\ell}{n_\ell + n_r} , \ \
	P_{\phi,q}((\mu,b]) \ = \ \frac{n_r}{n_\ell + n_r} .
\end{equation}
\end{Lemma}

In view of Lemmas~\ref{lem:exotic1} and \ref{lem:exotic2}, when searching for a MLE one should first check the numbers $\mu' := \max\{L_1,\ldots,L_n\}$ and $\mu'' := \min\{R_1,\ldots,R_n\}$. If $\mu' < \mu''$, then any pair $(\phi,q) \in \Theta$ such that $P_{\phi,q}((\mu',\mu'']) = 1$ is a MLE. For instance, one could just take $q = 0$ and the linear log-density
\[
	\phi(x) \ := \ \begin{cases}
		\alpha + \beta x & \text{for} \ x \in [\mu',\mu''] \\
		- \infty         & \text{for} \ x \in \R \setminus [\mu',\mu'']
	\end{cases}
\]
with arbitrary $\beta \in \R$ ($\beta < 0$ in case of $\mu'' = \infty$) and a suitable $\alpha = \alpha(\beta,\mu'\mu'') \in \R$.

In case of $\mu' = \mu'' =: \mu$, one has to check whether $L_i = R_i = \mu$ for at least one index $i$. If yes, there exists no MLE. If no, one has to determine the numbers $n_\ell, n_r$ and boundaries $a < \mu < b$ as described in Lemma~\ref{lem:exotic2}. Then any $(\phi,q) \in \Theta$ satisfying \eqref{eq:just two intervals} is a MLE. Here one can also show that $q = 0$ and 
\[
	\phi(x) \ := \ \begin{cases}
		\alpha + \beta x & \text{for} \ x \in [a,b] \\
		- \infty         & \text{for} \ x \in \R \setminus [a,b]
	\end{cases}
\]
with suitable $\alpha,\beta \in \R$ fulfill this constraint.

In case of at least one uncensored observation we have to rule out an additional pathological case:

\begin{Lemma}
\label{lem:exotic3}
Suppose that $L_{i_o} = R_{i_o} = \mu$ for some index $i_o$ and $\mu \in \R$. Further suppose that all observations satisfy $\mu \in [L_i,R_i]$ or $R_i = \infty$. Then for any $q \in (0,1)$,
\[
	\sup_{\phi \in \Phi(q)} \, \ell(\phi,q) \ = \ \infty .
\]
\end{Lemma}

\paragraph{Shape of the maximizers.}
We start this section with a rather simple and intuitive fact about the domains of $\hat{\phi}$ and $\hat{\phi}_0$, where the domain of a concave function $\phi$ is defined as
\[
	\dom(\phi) \ := \ \{x \in \R : \phi(x) > -\infty\} .
\]

\begin{Lemma}
\label{lem:basic}
Let $-\infty \le a < b \le \infty$ such that $\til{X}_1, \til{X}_2, \ldots, \til{X}_n \subset [a,b]$. If a MLE $\hat{\phi}_0$ exists, then $\dom(\hat{\phi}_0) \subset [a,b]$. If a MLE $(\hat{\phi},\hat{q})$ exists, then $\dom(\hat{\phi}) \subset [a,b]$.
\end{Lemma}

In what follows let
\[
	-\infty < \tau_1 < \tau_2 < \cdots < \tau_m < \tau_{m+1} = \infty
\]
such that
\[
	\{\tau_1,\tau_2,\ldots,\tau_{m+1}\}
	\ = \ \{L_1,L_2,\ldots,L_n\} \cup \{R_1,R_2,\ldots,R_n\} \cup \{\infty\} .
\]
In particular, $\tau_1 = \min\{L_1,L_2,\ldots,L_n\}$. We assume that $m \ge 2$, because otherwise Lemma~\ref{lem:exotic1}, \ref{lem:exotic2} or \ref{lem:exotic3} would apply. It follows directly from Lemma~\ref{lem:basic} that
\begin{align*}
	\dom(\hat{\phi}_0) \
	&\subset \ \begin{cases}
		[\tau_1,\infty) , \\
		[\tau_1,\tau_m] &\text{if} \ R_i < \infty \ \text{for} \ 1 \le i \le n ,
	\end{cases} \\
	\dom(\hat{\phi}) \
	&\subset \ [\tau_1,\infty) .
\end{align*}
One may even require that $\dom(\hat{\phi}) \subset [\tau_1,\tau_m]$, because for any $(\phi,q) \in \Theta$, the value of $\ell(\phi,q)$ remains the same if we replace $q$ with $q + \int_{\tau_m}^\infty e^{\phi(t)} \, dt$ and redefine $\phi(t) := -\infty$ for $t > \tau_m$.

Note that $\phi$ enters $\ell(\phi,q)$ only via the values $\phi(X_i)$ for those $i$ with $L_i = R_i$ and via the integrals $\int_{\tau_j}^{\tau_{j+1}} e^{\phi(t)} \, dt$, $1 \le j \le m$. Indeed we may restrict our attention to piecewise linear functions $\phi$ with at most $m-1$ changes of slope within their domain:

\begin{Theorem}[Shape of maximizers]
\label{thm:shape}
Let $(\phi,q) \in \Theta$ with $\ell(\phi,q) > -\infty$ and $\dom(\phi) \subset [\tau_1,\infty)$. Then there exists a $(\tilde{\phi},q) \in \Theta$ satisfying $\ell(\tilde{\phi},q) \ge \ell(\phi,q)$ and $\dom(\tilde{\phi}) \subset [\tau_1,\infty)$ and the following conditions:

\noindent
\textbf{(i)} \ For $j \in \{1,2,\ldots,m\}$ either $\dom(\tilde{\phi}) \cap (\tau_j,\tau_{j+1}) = \emptyset$ or $(\tau_j,\tau_{j+1}) \subset \dom(\tilde{\phi})$. In the latter case, $\tilde{\phi}$ is piecewise linear on $[\tau_j,\tau_{j+1}] \cap \R$ with at most one change of slope within $(\tau_j,\tau_{j+1})$. It is even linear on $[\tau_j,\tau_{j+1}] \cap \R$ if
\[
	\begin{cases}
		j \in \{1,m-1,m\} , \\
		j \ge 2 \ \ \text{and} \ \
			\dom(\tilde{\phi}) \subset [\tau_{j},\infty) , \\
		j \le m-2 \ \ \text{and} \ \
			\dom(\tilde{\phi}) \subset (-\infty,\tau_{j+1}] , \\
		(\tau_j,\tau_{j+1}) \subset \R \setminus \bigcup_{i=1}^n [L_i,R_i] .
	\end{cases}
\]

\noindent
\textbf{(ii)} \ Suppose that for indices $1 \le j < \ell \le m+1$ with $\ell - j \ge 2$,
\[
	\tau_j \in \dom(\tilde{\phi}) \subset (-\infty,\tau_\ell] ,
	\quad
	\int_{\tau_j}^{\tau_\ell} e^{\tilde{\phi}(x)} \, dx \ > \ 0
	\quad\text{and}\quad
	\tau_k \ \not\in \ \{R_1,\ldots,R_n\} \ \ \text{if} \ j < k < \ell .
\]
Then $\tilde{\phi}$ is linear on $[\tau_j,\tau_\ell] \cap \R$.

\noindent
\textbf{(iii)} \ Suppose that for indices $1 \le j < \ell \le m$ with $\ell - j \ge 2$,
\[
	\tau_\ell \in \dom(\tilde{\phi}) \subset \ [\tau_j,\infty) ,
	\quad
	\int_{\tau_j}^{\tau_\ell} e^{\tilde{\phi}(x)} \, dx \ > \ 0
	\quad\text{and}\quad
	\tau_k \ \not\in \ \{L_1,\ldots,L_n\} \ \ \text{if} \ j < k < \ell .
\]
Then $\tilde{\phi}$ is linear on $[\tau_j,\tau_\ell]$.

\noindent
\textbf{(iv)} \ Suppose that $[\tau_{j-1},\tau_{j+1}] \subset \dom(\tilde{\phi})$ for an index $2 \le j \le m-1$. Then $\tilde{\phi}$ has at most one change of slope within $(\tau_{j-1}, \tau_{j+1})$.
\end{Theorem}

\paragraph{Approximating the parameter spaces.}
In view of Theorem~\ref{thm:shape} we consider arbitrary tuples $\bs{t} = (t_1,t_2,\ldots,t_N)$ with $N \ge 2$ components $-\infty < t_1 < t_2 < \ldots < t_N < t_{N+1} := \infty$ and define
\begin{align*}
	\Phi_{\bs{t}}(q)
	\ := \ \bigl\{ \phi \in \Phi(q) : \
	&\dom(\phi) = [t_1,\infty) \ \ \text{and} \\
	&\phi \ \text{is linear on} \ [t_1,t_2], [t_2,t_3], \ldots, [t_{N-1},t_N], [t_N,\infty)
		\bigr\} , \\
	\Phi_{\bs{t}}^o(q)
	\ := \ \bigl\{ \phi \in \Phi(q) : \
	&\dom(\phi) = [t_1,t_N] \ \ \text{and} \\
	&\phi \ \text{is linear on} \ [t_1,t_2], [t_2,t_3], \ldots, [t_{N-1},t_N]
		\bigr\} .
\end{align*}
Note that functions $\phi \in \Phi_{\bs{t}}(q)$ and $\phi^o \in \Phi_{\bs{t}}^o(q)$ are completely determined by the tuples
\[
	\bigl( \phi(t_1), \phi(t_2), \ldots, \phi(t_N), \phi'(t_N\,+) \bigr)
	\quad\text{and}\quad
	\bigl( \phi^o(t_1), \phi^o(t_2), \ldots, \phi^o(t_N) \bigr) .
\]
In addition we need the larger sets $\overline{\Phi}_{\bs{t}}(q)$ and $\overline{\Phi}_{\bs{t}}^o(q)$ of functions in $\Phi(q)$ which may be represented as pointwise limits of sequences in $\Phi_{\bf{t}}(q)$ and $\Phi_{\bs{t}}^o(q)$, respectively. One can easily verify that
\begin{align*}
	\overline{\Phi}_{\bf{t}}(q) \
	&= \ \bigcup_{1 \le a \le N} \Phi_{(t_a,\ldots,t_N)}(q) \cup
		\bigcup_{1 \le a < b \le N} \Phi_{(t_a,\ldots,t_b)}^o(q) , \\
	\overline{\Phi}_{\bf{t}}^o(q) \
	&= \ \bigcup_{1 \le a < b \le N} \Phi_{(t_a,\ldots,t_b)}^o(q) .
\end{align*}

In case of $m \le 3$, when maximizing $\ell(\cdot)$ over $\Phi(0)$, we may replace $\Phi(0)$ with its subset $\overline{\Phi}_{(\tau_1,\ldots,\tau_m)}(0)$. To maximize $\ell(\cdot,\cdot)$ over $\Theta$, it suffices to consider the set $\overline{\Phi}_{(\tau_1,\ldots,\tau_m)}^o(q)$ in place of $\Phi(q)$ for $0 \le q < 1$.

In case of $m \ge 4$, our target functions $\hat{\phi}_0$ or $\hat{\phi}$ may contain knots in $\R \setminus \{\tau_1,\tau_2,\ldots,\tau_m\}$. Precisely, if we exclude the special situations described by Lemmas~\ref{lem:exotic1}, \ref{lem:exotic2} and \ref{lem:exotic3}, then there exist a smallest index $j_1 \in \{2,\ldots,m\}$ such that $\tau_{j_1} \in \{L_1,L_2,\ldots,L_n\}$ and a largest index $j_2 \in \{1,2,\ldots,m-1\}$ such that $\tau_{j_2} \in \{R_1,R_2,\ldots,R_n\}$. By Theorem~\ref{thm:shape} (i), (ii) and (iii) we may focus on target functions that are linear on the part of their domain that lies before $\tau_{j_1}$ and on the part that lies after $\tau_{j_2}$. So in case of $j_1 \ge j_2$, it still suffices to consider $\overline{\Phi}_{(\tau_1,\tau_2,\ldots,\tau_m)}(0)$ instead of $\Phi(0)$ and, when maximizing of $\Theta$, to consider $\overline{\Phi}_{(\tau_1,\tau_2,\ldots,\tau_m)}^o(q)$ instead of $\Phi(q)$ for $q \in (0,1]$. If $j_1 < j_2$, however, we approximate $\Phi(q)$ with $\overline{\Phi}_{\bs{t}}(q)$ or $\overline{\Phi}_{\bs{t}}^o(q)$, where $\bs{t}$ contains $\tau_1,\tau_2,\ldots,\tau_m$ and a fine grid of extra points in $(\tau_j,\tau_{j+1})$ for each $j \in \{j_1, \ldots, j_2-1\}$ such that $(\tau_j,\tau_{j+1}) \not\subset \R \setminus \bigcup_{i=1}^n [L_i,R_i]$.

\section{Algorithms}
\label{sec:algorithms}

Throughout this section we exclude the special situations described in Lemmas~\ref{lem:exotic1}, \ref{lem:exotic2} and \ref{lem:exotic3}. In particular, we assume that $R_i < \infty$ for at least one observation, and $m \ge 2$.

\paragraph{Augmented log-likelihood functions.}
Using a trick of \citet{Silverman_1982}, we can remove the constraint \eqref{eq:constraint Theta}. Let $\Phi$ be the set of all concave and upper semicontinuous functions $\phi: \R \to [-\infty,\infty)$ such that $\phi(x) \to - \infty$ as $|x| \to \infty$. Define the augmented log-likelihood as
\begin{equation}
\label{eq:loglike_silverman}
	\Lambda(\phi,q) \ := \ \ell(\phi,q) - \int e^{\phi(x)} \, dx - q + 1
\end{equation} 
for $\phi \in \Phi$ and $q \ge 0$, and set $\Lambda(\phi) := \Lambda(\phi,0)$. In case of $\Lambda(\phi,q) > -\infty$,
\[
	\frac{\partial}{\partial c} \Lambda(\phi + c, q e^c)
	\ = \ 1 - e^c \Bigl( \int e^{\phi(x)} \, dx + q \Bigr)
\]
for any $c \in \R$. Hence for fixed $\phi \in \Phi$ and $q \ge 0$ such that $\Lambda(\phi,q) > -\infty$,
\[
	\argmax_{c \in \R} \, \Lambda(\phi + c, q e^c)
	\ = \ - \log \Bigl( \int e^{\phi(x)} \, dx + q \Bigr) .
\]
Moreover, for this particular value $c$, the parameter $(\phi+c,q e^c)$ belongs to $\Theta$, and $\ell(\phi+c, q e^c) = \Lambda(\phi + c, q e^c)$. These considerations imply the following result:

\begin{Lemma}
\label{lem:silverman}
\begin{equation*}
	\argmax_{(\phi,q) \in \Theta} \ell(\phi,q)
	\ = \ \argmax_{(\phi,q) \in \Phi \times [0,\infty)} \Lambda(\phi,q)
\end{equation*}
and 
\begin{equation*}
	\argmax_{\phi \in \Phi(0)} \ell(\phi)
	\ = \ \argmax_{\phi \in \Phi} \Lambda(\phi),
\end{equation*}
where $\argmax$ refers to the (possibly empty) \emph{set} of the corresponding maximizers.
\end{Lemma}

\paragraph{Optimizing the cure parameter.}
It seems that we cannot use Silverman's trick to maximize $\ell(\phi,q)$ for a given fixed value $q>0$. But the augmented likelihood $\Lambda(\phi,q)$ is useful for finding a better value of $q$: Let $\phi$ be a fixed function in $\Phi$ such that $\Lambda(\phi,q) > -\infty$ for some (and thus all) $q > 0$. Then
\[
	\frac{\partial}{\partial q} \, \Lambda(\phi,q)
	\ = \ \frac{1}{n} \sum_{i=1}^n
		\frac{1_{[R_i = \infty]}}{\int_{L_i}^{\infty} e^{\phi(x)} \, dx + q}
		- 1 .
\]
In the special case of all right endpoints $R_i$ being finite, $\Lambda(\phi,q) < \Lambda(\phi)$ for any $q > 0$. Otherwise, if $R_i = \infty$ for at least one observation, the right hand side is strictly decreasing in $q > 0$ and strictly negative for $q > \overline{q} := \#\{i : R_i = \infty\}/n$. Hence one can easily maximize $\Lambda(\phi,q)$ with respect to $q \ge 0$ as follows: If
\begin{equation}
\label{eq:optimal.q.0}
	\frac{1}{n} \sum_{i=1}^n \frac{1_{[R_i = \infty]}}{\int_{L_i}^\infty e^{\phi(x)} \, dx}
	\ \le \ 1 ,
\end{equation}
then the maximizer is given by $q = 0$. Otherwise it is the unique number $q \in (0,\overline{q}]$ such that
\begin{equation}
\label{eq:optimal.q}
	\frac{1}{n} \sum_{i=1}^n
		\frac{1_{[R_i = \infty]}}{\int_{L_i}^{\infty} e^{\phi(x)} \, dx + q}
	\ = \ 1 .
\end{equation}
This number may be determined, for instance, by binary search or a Newton procedure. Note also that $\overline{q} < 1$ by assumption.

\paragraph{The EM paradigm.}
Maximizing the augmented log-likelihood function $\Lambda(\phi,q)$ with respect to $\phi \in \Phi$ for a fixed value of $q \ge 0$ is a non-trivial task. A major problem is that $\ell(\cdot,q)$ is convex rather than linear or concave. Namely, let $\phi, \phi_{\rm new} \in \Phi$ with $\Lambda(\phi,q), \Lambda(\phi_{\rm new},q) > - \infty$ and $\dom(\phi_{\rm new}) = \dom(\phi)$. Further let $v(x) := \phi_{\rm new}(x) - \phi(x)$ for $x \in \dom(\phi)$ and $v(x) := 0$ otherwise. Then $[0,1] \mapsto \Lambda(\phi + tv, q)$ is continuous, and for $t \in (0,1]$,
\begin{align}
\label{eq:dir.deriv.1}
	\frac{d}{dt} \ell(\phi + tv, q) \
	&= \ \frac{1}{n} \sum_{i=1}^n \Bigl( 1_{[L_i = R_i]} v(X_i)
		+ 1_{[L_i < R_i]} \frac{\int_{L_i}^{R_i} v(x) e^{\phi(x) + t v(x)} \, dx}
			{\int_{L_i}^{R_i} e^{\phi(x) + t v(x)} \, dx + 1_{[R_i = \infty]} q} \Bigr) \\
	\nonumber
	&= \ \frac{1}{n} \sum_{i=1}^n
		\Ex_{\phi + tv,q} \bigl( v(Y) \,\big|\, Y \in \tilde{X}_i \bigr) , \\
\label{eq:dir.deriv.2}
	\frac{d^2}{dt^2} \ell(\phi + tv, q) \
	&= \ \frac{1}{n} \sum_{i=1}^n
		\Var_{\phi + tv,q} \bigl( v(Y) \,\big|\, Y \in \tilde{X}_i \bigr) \ \ge \ 0 ,
\end{align}
where $v(\infty) := 0$, the observations $\tilde{X}_i$ are viewed temporarily as fixed, and $Y$ denotes a random variable such that
\[
	\Pr_{\phi + tv,q}(Y \in B)
	\ = \ \frac{\int_{B\cap\R} e^{\phi(x) + tv(x)} \, dx + 1_{[\infty \in B]} q}
		{\int_{\R} e^{\phi(x) + tv(x)} \, dx + q}
\]
for Borel sets $B \subset (-\infty,\infty]$. We may also write
\[
	\frac{d}{dt} \Big|_{t=0}^{} \ell(\phi + tv, q)
	\ = \ \int v \, dM_{\phi,q}
\]
with the following sub-probability distribution $M_{\phi,q}$ on $\R$: For any Borel set $B \subset \R$,
\[
	M_{\phi,q}(B)
	\ := \ n_{}^{-1} \sum_{i=1}^n \biggl( 1_{[L_i = R_i]} 1_{[X_i \in B]}
		+ 1_{[L_i < R_i]} \frac{\int_{B \cap (L_i,R_i)} e^{\phi(x)} \, dx}
			     {\int_{(L_i,R_i)} e^{\phi(x)} \, dx + 1_{[R_i = \infty]} q} \biggr) .
\]
Now we propose to replace $\ell(\psi,q)$ in the definition of $\Lambda(\psi,q)$ with its linearization $\ell(\phi,q) + \int (\psi - \phi) \, dM_{\phi,q} = c(\phi,q) + \int \psi \, dM_{\phi,q}$ and to maximize
\[
	\Lambda_{\phi,q}(\psi)
	\ := \ \int \psi \, dM_{\phi,q} - \int e_{}^{\psi(x)} \, dx - q + 1
\]
over all $\psi \in \Phi$. Note also that
\[
	\int \psi \, dM_{\phi,q}
	\ = \ \Ex \bigl( \tilde{\ell}(\psi,0)
		\,\big|\, \tilde{X}_1, \tilde{X}_2, \ldots, \tilde{X}_n \bigr) ,
\]
i.e.\ the conditional expectation of the complete-data log-likelihood, given the available data. This is the more traditional motivation for the EM algorithm. Existence of a unique maximizer of $\Lambda_{\phi,q}(\cdot)$ over $\Phi$ is guaranteed by the following auxiliary result which is just a modification of Theorem~2.2 of \citet{Duembgen_etal_2011}:

\begin{Lemma}
\label{lem:DSS}
Let $M$ be a finite measure on the Borel subsets of $\R$ such that $S(M) := \bigl\{ x \in \R : 0 < M((-\infty,x]) < M(\R) \bigr\}$ is non-empty and $\int |x| \, M(dx) < \infty$. Then there exists a unique maximizer $\phi \in \Phi$ of
\[
	\int \phi \, dM - \int e^{\phi(x)} \, dx .
\]
This maximizer $\phi$ satisfies the equation $\int e^{\phi(x)} \, dx = M(\R)$, and the closure of $\dom(\phi)$ equals the closure of $S(M)$.
\end{Lemma}

Suppose our current candidate for $(\hat{\phi},\hat{q})$ is $(\phi,q)$, where either $q = 0$ or $q > 0$ satisfies \eqref{eq:optimal.q}. Then the measure $M_{\phi,q}$ satisfies $M_{\phi,q}(\R) = 1 - q$. Now let $\phi_{\rm new}$ be the maximizer of $\Lambda_{\phi,q}(\cdot)$ over $\Phi$. It will automatically satisfy the equation
\[
	\int e_{}^{\phi_{\rm new}(x)} \, dx
	\ = \ 1 - q ,
\]
so $\phi_{\rm new} \in \Phi(q)$. Moreover,
\[
	\Lambda(\phi_{\rm new},q) \ > \ \Lambda(\phi,q)
	\quad\text{unless} \ \phi_{\rm new} \equiv \phi .
\]
For if $\phi_{\rm new} \ne \phi$, then the definition of $\phi_{\rm new}$ and convexity of $\ell(\cdot,q)$ imply that
\begin{align*}
	0 \
	&< \ \Lambda_{\phi,q}(\phi_{\rm new}) - \Lambda_{\phi,q}(\phi) \\
	&= \ \int (\phi_{\rm new} - \phi) \, dM_{\phi,q}
		- \int e_{}^{\phi_{\rm new}(x)} \, dx
		+ \int e_{}^{\phi(x)} \, dx \\
	&= \ \frac{d}{dt} \Big|_{t=0}^{} \ell(\phi + t(\phi_{\rm new} - \phi), q)
		- \int e_{}^{\phi_{\rm new}(x)} \, dx
		+ \int e_{}^{\phi(x)} \, dx \\
	&\le \ \ell(\phi_{\rm new},q) - \ell(\phi,q)
		- \int e_{}^{\phi_{\rm new}(x)} \, dx
		+ \int e_{}^{\phi(x)} \, dx \\
	&= \ \Lambda(\phi_{\rm new},q) - \Lambda(\phi,q) .
\end{align*}

Now we replace $\phi$ with $\phi_{\rm new}$. When maximizing $\ell(\cdot,\cdot)$ over $\Theta$, we also recalculate $q$ via \eqref{eq:optimal.q.0} and \eqref{eq:optimal.q}. This yields possibly a further increase of $\Lambda(\phi,q)$, and the new value $q$ equals $0$ or satisfies \eqref{eq:optimal.q}. In principle this procedure is iterated until the ``difference'' between $\phi$ and $\phi_{\rm new}$ becomes negligible.

\paragraph{Practical implementation of the EM step.}
Maximization of $\Lambda_{\phi,q}(\cdot)$ over $\Phi$ may be achieved via an active set algorithm as described in \citet{Duembgen_etal_2007} if we approximate $\Phi$ by finite-dimensional sets $\Phi_{\bs{t}}$ or $\Phi_{\bs{t}}^o$. The latter two are defined as the sets $\Phi_{\bs{t}}(q)$ and $\Phi_{\bs{t}}^o(q)$ in Section~\ref{sec:parameters} with the constraint $\phi \in \Phi(q)$ replaced with the requirement $\phi \in \Phi$. Initially the tuple $\bs{t} = (t_1,t_2,\ldots,t_N)$ is chosen as described at the end of Section~\ref{sec:parameters}. Later on it may be a subtuple of that.

Suppose that $\phi$ is a log-density in $\Phi_{\bs{t}}(q)$ or $\Phi_{\bs{t}}^o(q)$ where either $q = 0$ or $q \in (0,1)$ satisfies \eqref{eq:optimal.q}. Since $[t_1,t_N] \subset [\tau_1,\tau_m]$, $\dom(\phi)$ is a closed set and equal to the convex hull of the support of $M_{\phi,q}$. Hence the closure of the domain of $\argmax_{\psi \in \Phi} \Lambda_{\phi,q}(\psi)$ is equal to $\dom(\phi)$. Consequently, if we restrict our attention to candidates in $\overline{\Phi}_{\bs{t}}$, then it even suffices to consider functions in
\[
	\tilde{\Phi} \ := \ \begin{cases}
		\Phi_{\bs{t}}   &\text{if} \ \phi \in \Phi_{\bs{t}}(q) , \\
		\Phi_{\bs{t}}^o &\text{if} \ \phi \in \Phi_{\bs{t}}^o(q) .
	\end{cases}
\]
But for $\psi \in \tilde{\Phi}$ we may write
\[
	\int \psi \, dM_{\phi,q}
	\ = \ \sum_{j=1}^{N} w_j \psi(t_j)
		+ w_{N+1} \psi'(t_N\,+) ,
\]
where $\psi'(t_N\,+) := -\infty$ in case of $\tilde{\Phi} = \Phi_{\bs{t}}^o$, and
\begin{align*}
	w_1 \ &:= \ w_{1,c} + w_{1,r} , \\
	w_j \ &:= \ w_{j,\ell} + w_{j,c} + w_{j,r} \quad \text{for} \ 2 \le j < N , \\
	w_N \ &:= \ w_{N,\ell} + w_{N,c} + M_{\phi,q}((t_N,\infty)) , \\
	w_{j,c} \    &:= \ \#\{i : L_i = R_i = t_j\}/n , \\
	w_{j,\ell} \ &:= \ \int_{t_{j-1}}^{t_j}
		\frac{x - t_{j-1}}{t_j - t_{j-1}} \, M_{\phi,q}(dx)
		\quad \text{for} \ 2 \le j \le N , \\ 
	w_{j,r} \    &:= \ \int_{t_j}^{t_{j+1}}
		\frac{t_{j+1} - x}{t_{j+1} - t_j} \, M_{\phi,q}(dx)
		\quad \text{for} \ 1 \le j < N , \\
	w_{N+1} \    &:= \ \int_{t_N}^\infty (x - t_N) \, M_{\phi,q}(dx) .
\end{align*}
(Note that $w_{N+1} = 0$ in case of $\phi \in \Phi_{\bs{t}}^o(q)$.) Hence $\int \psi \, dM_{\phi,q}$ is a simple linear combination of $\bigl( \psi(t_1), \psi(t_2), \ldots, \psi(t_N), \psi'(t_N\,+) \bigr)$. The second part of $\Lambda_{\phi,q}(\psi)$ may be written as
\[
	\int e^{\psi(x)} \, dx
	\ = \ \sum_{j=1}^{N-1} (t_{j+1} - t_j) J \bigl( \psi(t_j),\psi(t_{j+1}) \bigr)
		+ \tilde{J} \bigl( \psi(t_N), \psi'(t_N\,+) \bigr) ,
\]
where for $a,b \in \R$ and $c \in [-\infty,0)$,
\begin{align*}
	J(a,b) \
	&:= \ \int_0^1 \exp((1 - t)a + tb) \, dt
		\ = \ \begin{cases}
			\exp(a) & \text{if} \ a = b , \\
			(\exp(b) - \exp(a))/(b - a) & \text{if} \ a \ne b ,
		\end{cases} \\
	\tilde{J}(a,c) \
	&:= \ \int_0^\infty \exp(a + ct) \, dt
		\ = \ \exp(a)/(-c) .
\end{align*}

\paragraph{Stopping the EM iterations and modifying the domains of $\hat{\phi}_0$ or $\hat{\phi}$.}
Let $\phi_1, \phi_2, \phi_3, \ldots$ be our candidates for $\hat{\phi}_0$ or $\hat{\phi}$. One should stop iterating the EM step (plus optimization with respect to $q$) when the changes in the (sub-)probability density $f_k = \exp(\phi_k)$ become negligible. A reasonable distance measure would be the $L^1$-distance $\int \bigl| f_{k+1}(x) - f_k(x) \bigr| \, dx$, but the following upper bound is much easier to compute and of the same order of magnitude:
\begin{align*}
	\int \bigl| f_{k+1}(x) - f_k(x) \bigr| \, dx
	\ \le \ &\sum_{j=1}^{N-1} (t_{j+1} - t_j)
			\bigl( J(\overline{m}_{j}, \overline{m}_{j+1})
				- J(\underline{m}_{j}, \underline{m}_{j+1}) \bigr) \\
		&+ \ \exp(\overline{m}_N) \, \overline{m}_{N+1}
			- \exp(\underline{m}_N) \, \underline{m}_{N+1} ,
\end{align*}
where $\underline{m}_{j}$ and $\overline{m}_{j}$ are the minimum and maximum, respectively, of $\bigl\{ \phi_k(t_j), \phi_{k+1}(t_j) \bigr\}$ if $j \le N$ and of $\bigl\{ |\phi_k'(t_N\,+)|^{-1},|\phi_{k+1}'(t_N\,+)|^{-1} \bigr\}$ if $j = N+1$.

It happens often that $\phi_k \to -\infty$ on a non-empty subset of $\dom(\phi_1)$, which may lead to numerical problems or a waste of computation time. One possible way out is as follows: For the computation of $\hat{\phi}_0$ one could replace $\ell(\phi)$ with
\[
	\frac{\eps_1 \phi(\tau_1) + n \ell(\phi)
		+ \eps_2 \log \bigl( \int_{\tau_m}^\infty e^{\phi(x)} \, dx \bigr)}{\eps_1 + n + \eps_2}
\]
with certain numbers $0 \le \eps_1, \eps_2 \ll 1$, where $\eps_1 > 0$ unless $L_i = R_i = \tau_1$ for some index $i$, and $\eps_2 > 0$ unless $(L_j,R_j] = (\tau_m,\infty]$ for some index $j$. For the computation of $(\hat{\phi},\hat{q})$ and working with $\Phi_{\bs{t}}^o$, we may replace $\ell(\phi,q)$ with
\[
	\frac{\eps_1 \phi(\tau_1) + n \ell(\phi,q) + \eps_2 \phi(\tau_m)}{\eps_1 + n + \eps_2} ,
\]
where $\eps_1$ is chosen as before, while $\eps_2 > 0$ unless $L_j = R_j = \tau_m$ for some index $j$. Hence we add artifical ``observations'' $\{\tau_1\}$ and $(\tau_m,\infty]$ or $\{\tau_m\}$ with very small weights to our original data set.

One can be more ambitious and try to estimate the domain of $\hat{\phi}_0$ or $\hat{\phi}$. That means, whenever there is strong evidence for $\dom(\phi_1)$ being too large, the candidate set $\tilde{\Phi}$ may be reduced as follows:

\noindent
\textsl{Possible reduction~1.} \
Suppose that we would like to compute $\hat{\phi}_0$ and that $\tilde{\Phi} = \Phi_{\bs{t}}$. With $\theta := \phi_k'(t_N\,+)$ we may write
\begin{align*}
	\frac{\partial}{\partial \theta} \Lambda(\phi_k) \
	&= \ \frac{\partial}{\partial \theta} \biggl(
		\frac{1}{n} \sum_{i=1}^n 1_{[R_i = \infty]}
			\log \Bigl( P_{\phi_k}((L_i,t_{N-1}]) + \frac{\exp(\phi_k(t_N))}{-\theta} \Bigr)
		- \frac{\exp(\phi_k(t_N))}{-\theta} \biggr) \\
	&= \ \frac{\exp(\phi_k(t_N))}{\theta^2} \Bigl(
		\frac{1}{n} \sum_{i=1}^n
			\frac{1_{[R_i = \infty]}}{P_{\phi_k}((L_i,t_{N-1}]) + \exp(\phi_k(t_N))/(-\theta)}
		- 1 \Bigr) .
\end{align*}
The latter is strictly negative for all values of $\theta \in (-\infty,0)$ if
\begin{equation}
\label{eq:remove.tN+1}
	\frac{1}{n} \sum_{i=1} \frac{1_{[R_i = \infty]}}{P_{\phi_k}((L_i,t_{N-1}])}
	\ \le \ 1 .
\end{equation}
In that case, and if $\int \bigl| f_k(x) - f_{k-1}(x) \bigr| \, dx$ is below some prespecified threshold, we recompute $\phi_k$, working from now on with $\tilde{\Phi} = \Phi_{\bs{t}}^o$.

\noindent
\textsl{Possible reduction~2.} \
Suppose that $\tilde{\Phi} = \Phi_{\bs{t}}^o$ and $\#\{i : L_i = R_i = t_N\} = 0$. With $\gamma := \phi_k(t_{N-1})$, $\theta := \phi_k(t_N)$ and $\delta := t_N - t_{N-1}$ we may write
\begin{align*}
	\frac{\partial}{\partial \theta} \Lambda(\phi_k,q_k) \
	&= \ \frac{\partial}{\partial \theta} \biggl(
		\frac{1}{n} \sum_{i=1}^n 1_{[L_i < t_N \le R_i]}
			\log \bigl( \pi_{ki} + \delta J(\gamma,\theta) \bigr)
		- \delta J(\gamma,\theta) \biggr) \\
	&= \ \delta J_{01}(\gamma,\theta) \Bigl(
		\frac{1}{n} \sum_{i=1}^n
			\frac{1_{[L_i < t_N \le R_i]}}
			     {\pi_{ki} + \delta J(\gamma,\theta)}
		- 1 \Bigr) ,
\end{align*}
where
\[
	\pi_{ki} \ := \ \int_{L_i}^{t_{N-1}} e_{}^{\phi_k(x)} \, dx
		+ 1_{[R_i = \infty]} q_k
\]
and $J_{01}(a,b) := \int_0^1 \exp((1 - t)a + tb) t \, dt > 0$. Consequently, if $\int \bigl| f_k(x) - f_{k-1}(x) \bigr| \, dx$ is below some prespecified threshold and if
\begin{equation}
\label{eq:remove.tN}
	\frac{1}{n} \sum_{i=1}
	\frac{1_{[L_i < t_N \le R_i]}}{\pi_{ki}}
	\ \le \ 1 ,
\end{equation}
then we recompute $\phi_k$, working from now on with $\tilde{\Phi} = \Phi_{(t_1,\ldots,t_{N-1})}^o$.

\noindent
\textsl{Possible reduction~3.} \
Analogously, suppose that $\tilde{\Phi} = \Phi_{\bs{t}}^o$ and $\#\{i : L_i = R_i = t_1\} = 0$. With $\theta := \phi_k(t_1)$, $\gamma := \phi_k(t_2)$ and $\delta := t_2 - t_1)$ we may write
\[
	\frac{\partial}{\partial \theta} \Lambda(\phi_k,q_k)
	\ = \ \delta J_{01}(\gamma,\theta) \Bigl(
		\frac{1}{n} \sum_{i=1}^n
			\frac{1_{[L_i = t_1]}}
			     {\pi_{ki} + \delta J(\gamma,\theta)}
		- 1 \Bigr) ,
\]
where
\[
	\pi_{ki} \ := \ \int_{t_2}^{R_i} e_{}^{\phi_k(x)} \, dx
		+ 1_{[R_i = \infty]} q_k .
\]
Hence if $\int \bigl| f_k(x) - f_{k-1}(x) \bigr| \, dx$ is below some prespecified threshold and if
\begin{equation}
\label{eq:remove.t1}
	\frac{1}{n} \sum_{i=1}
	\frac{1_{[L_i = t_1]}}{\pi_{ki}}
	\ \le \ 1 ,
\end{equation}
then we recompute $\phi_k$, working from now on with $\tilde{\Phi} = \Phi_{(t_2,\ldots,t_N)}^o$.

\section{Identifiability and Consistency}
\label{sec:identifiability.consistency}

\paragraph{Partial identifiability of $q$.}
Without any shape constraints on $\phi = \log f$, the cure parameter $q$ would not be identifiable. Indeed, with $\tau_*$ denoting the maximum of $\{L_1,\ldots,L_n\} \cup \{R_1,\ldots,R_n\} \cap \R$, the data $\til{X}_1, \ldots, \til{X}_n$ would only provide information about $P = P_{\phi,q}$ on $(-\infty,\tau_*]$ and the number $P((\tau_*,\infty])$. Even if we knew the distribution function $F$ of $P$ on the whole interval $[-\infty,\tau_*]$, we could only conclude that
\[
	q \ \in \ [ 0, 1 - F(\tau_*)] .
\]

On the other hand, let $\phi$ be concave, and suppose we know $F$ only on some bounded interval $(a,b)$ with $F(b\,-) > 0$. Then we know $F(b)$, $f(b)$ and $\phi'(b\,-)$ as well, and the unknown parameter $q = 1 - F(b) - \int_0^\infty f(b+s) \, ds$ satisfies
\[
	q \ \begin{cases}
		\le \ 1 - F(b) , \\
		= \ 1 - F(b)
			& \text{if} \ f(b) = 0 , \\
		\ge \ 1 - F(b) - f(b)/|\phi'(b\,-)|
			& \text{if} \ f(b) > 0 > \phi'(b\,-) .
	\end{cases}
\]
The latter inequality follows from $f(b+s) \le f(b) \exp(\phi'(b\,-) s)$ for arbitrary $s \ge 0$. Hence if $b$ is sufficiently large, we get an equality or at least nontrivial lower and upper bounds for $q$.

\paragraph{Consistency.}
For simplicity we restrict our attention to the setting of interval-censoring: Let $P = P_{\phi,q}$ and $\hat{P}_n = P_{\hat{\phi},\hat{q}}$, where $(\hat{\phi},\hat{q}) = (\hat{\phi}_n,\hat{q}_n)$ is based on the following observations: For $1 \le i \le n$ let $A =: T_{n,i,0} < T_{n,i,1} < \cdots < T_{n,i,M_{ni}} < T_{n,i,M_{ni}+1} := \infty$ be given design points, where $A \in [-\infty,\infty)$ is a known lower bound for the support of $P$. Then observation $\til{X}_i = \til{X}_{n,i}$ is defined as the unique interval $\TT_{n,i,j} = (T_{n,i,j-1},T_{n,i,j}]$, $1 \le j \le M_{ni}+1$, containing $X_i$.

For instance, in connection with event times, $A = 0$ and $T_{n,i,1}, \ldots, T_{n,i,M_{ni}}$ could be inspection time points at which one determines whether the event in question has already happened or not.

In this setting, Theorem~\ref{thm:existence} guarantees existence of a MLE $\hat{P}_n$. The following consistency result is essentially Theorem~3 of \cite{Duembgen_Freitag_Jongbloed_2006} with obvious modifications of its proof. Throughout this section asymptotic statements refer to $n \to \infty$.

\begin{Theorem}[Consistency for interval-censored data]
\label{thm:consistency}
If $n^{-1} \sum_{i=1}^n M_{ni}^\gamma = O(1)$ for some $\gamma > 1/2$, then
\[
	\frac{1}{n} \sum_{i=1}^n \sum_{j=1}^{M_{ni}+1} \bigl| (\hat{P}_n - P)(\TT_{n,i,j}) \bigr|
	\ \to_p^{} \ 0 .
\]
\end{Theorem}

Starting from this general result one can obtain more traditional consistency statements under additional assumptions on the time points $T_{n,i,j}$. In what follows let $F, \hat{F}_n : [-\infty,\infty] \to [0,1]$ be the distribution functions of $P$ and $\hat{P}_n$, respectively. Furthermore let
\[
	H_n(B) \ := \ \frac{1}{n} \sum_{i=1}^n \max_{j=1,\ldots,M_{ni}} 1_{[T_{n,i,j} \in B]}
\]
for $B \subset \R$. Then Theorem~\ref{thm:consistency} implies the following result:

\begin{Corollary}[Consistency for interval-censored data]
\label{cor:consistency}
Let $n^{-1} \sum_{i=1}^n M_{ni}^\gamma = O(1)$ for some $\gamma > 1/2$.

\noindent
\textbf{(i)} \ Suppose that $\liminf_{n \to \infty} H_n([x,y]) > 0$ for two real numbers $x \le y$. Then
\[
	\hat{F}_n(x) \le F(y) + o_p(1)
	\ \ \text{and} \ \ \hat{F}_n(y) \ge F(x) + o_p(1) .
\]

\noindent
\textbf{(ii)} \ Let $-\infty \le a < b \le \infty$ such that $\liminf_{n \to \infty} H_n((x,y)) > 0$ whenever $a \le x < y \le b$. Then
\[
	\hat{F}_n(x) \ \to_p \ F(x)
\]
for any $x \in (a,b)$.
\end{Corollary}

Suppose, for instance, that $M_{ni} = 1$ for all $n$ and $i$. A special example for this setting is current status data. Here $H_n$ is the empirical distribution of the time points $T_{n,i,1}$, $1 \le i \le n$. If $H_n$ converges weakly to a probability distribution $H$ on the real line such that the distribution function of $H$ is strictly increasing on an open interval $(a,b) \subset \R$, then the assumption of Corollary~\ref{cor:consistency}, part~(ii) is satisfied.

The subsequent result is no longer restricted to the setting of purely interval-censored data. It shows that pointwise stochastic convergence of $\hat{F}_n$ to $F$ on a nondegenerate interval $(a,b)$ implies uniform convergence in probability, unless $F$ is constant on $(a,b)$. Furthermore, the corresponding estimator $\hat{f}_n$ of the density $f = e^\phi$ is consistent on $(a,b)$, too, and the estimator $\hat{q}_n$ of $q$ satisfies certain inequalities. In what follows we denote the positive part of a real number $s$ by $s^+ := \max\{s,0\}$.

\begin{Theorem}[Weak implies strong convergence]
\label{thm:consistency2}
Let $-\infty \le a < b \le \infty$ such that $F(b\,-) > F(a)$ and $\hat{F}_n(x) \to_p F(x)$ for any fixed $x \in (a,b)$. For $\delta > 0$ let $C_\delta$ be the set of all real $x \in [a,b]$ such that $f(x) > 0$ or $f$ is continuous on $(x \pm \delta)$. Further let $D_\delta$ be the set of all real $x \in [a,b]$ such that $1_{(a,b)} f$ is continuous on $(x \pm \delta)$. Then for any fixed $\delta > 0$,
\[
	\sup_{x \in C_\delta} \bigl( \hat{f}_n(x) - f(x) \bigr)^+
	\ \to_p \ 0
	\quad\text{and}\quad
	\sup_{x \in D_\delta} \bigl| \hat{f}_n(x) - f(x) \bigr|
	\ \to_p \ 0 .
\]
Moreover,
\[
	\int_a^b \bigl| \hat{f}_n(x) - f(x) \bigr| \, dx
	\ \to_p \ 0
	\quad\text{and}\quad
	\sup_{x \in [a,b]} \bigl| \hat{F}_n(x) - F(x) \bigr| \ \to_p \ 0 .
\]
Finally, if $b = \infty$, then $\hat{q}_n \to_p q$. Otherwise
\[
	\hat{q}_n \ \begin{cases}
		\le \ 1 - F(b) + o_p(1) , \\
		= \ 1 - F(b) + o_p(1)
			& \text{if} \ f(b) = 0 , \\
		\ge \ 1 - F(b) - f(b)/|\phi'(b\,-)| + o_p(1)
			& \text{if} \ f(b) > 0 > \phi'(b\,-) .
	\end{cases}
\]
\end{Theorem}

The statements about $\hat{f}_n - f$ in this theorem are similar to results of \cite{Schuhmacher_Huesler_Duembgen_2009a} in the context of log-concave probability densities on $\R^d$. They imply that
\[
	\hat{f}_n(x) \ \to_p \ f(x)
\]
for any $x \in (a,b)$ at which $f$ is continuous, and
\[
	\hat{f}_n(y) \ \le \ f(y) + o_p(1)
\]
for any real $y \in [a,b]$.

\section{Applications}
\label{sec:examples}

The algorithm described in the previous section was implemented and made publicly available as contributed package {\tt logconcens} \citep{Schuhmacher_etal_2013} for the statistical computing environment R \citep{R_2013}. We give here two demonstrations of this implementation, one for simulated interval-censored data and one for real right-censored data. In both cases we used the domain reduction technique detailed above, but the trick of adding artificial very small or large pseudo-observations with little weights led virtually to the same densities and survival functions.

\paragraph{Simulated data.}

We simulate event times $X_i$, $1 \leq i \leq 100$, from a $\Gamma(3,1)$-distribution and inspect them according to independent homogeneous Poisson processes with rate $1$. The latter means that for each $i$, we consider a random sequence $(T_{i,j})_{j=1}^\infty$ which is independent from $X_i$, starts at $T_{i,1} = 0$ and has independent, standard exponentially distributed increments $T_{i,j} - T_{i,j-1}$, $j \ge 2$. Then $\til{X}_i$ is the unique interval $(T_{i,j-1},T_{i,j}]$, $j \ge 2$, containing $X_i$.

Figure~\ref{fig:simulated} presents the generated data and comparisons of our log-concave NPMLE in terms of log-densities and survival functions. In the first panel we see the censored data consisting of the $n$ intervals $\tilde{X}_i$ sorted by their left endpoints. The second panel compares our estimator $\hat{\phi}$ to the true log-density of the Gamma distribution and to the NPMLE based on the exact data $X_i$. The differences are rather small. Note that $\hat{q} = 0$ and $\hat{\phi} = \hat{\phi}_0$ because all right endpoints $R_i$ are finite. The third panel compares the survival function $\hat{S}$ obtained from $\hat{\phi}$ to the true survival function $S$ and to the unconstrained nonparametric maximum likelihood estimator of the survival function from \cite{Turnbull_1976} (produced with the R package {\tt interval}; see \citealp{Fay_2013}). Compared to the latter the survival curve stemming from $\hat{\phi}$ is clearly preferable as it captures not only the approximate course but also the smoothness of the true survival curve.

In order to analyze the performance of the estimators more thoroughly, we simulated 500 data sets by the above procedure and computed $\hat{\phi}$ and $\hat{S}$ every time. The average supremum norm of $|\hat{S} - S|$ was 0.0614, which compares favourably to the value of 0.1540 obtained for the same quantity if we replace $\hat{S}$ with the Turnbull estimate.

To study the performance for a distribution with a positive cure parameter, we also simulated 500 data sets $(X_i$, $1 \leq i \leq 100)$ from the distribution $0.7\cdot\Gamma(3,1) + 0.3\cdot\delta_{\infty}$ and inspected them according to Poisson processes that were restricted to only six inspection times each. The average supremum norm of $|\hat{F} - F| = |\hat{S} - S|$ was then 0.0763 and the average estimation error $|\hat{q} - q|$ for the cure parameter was 0.0514. Replacing $\hat{S}$ and $\hat{q}$ with the Turnbull estimate and its rightmost value, we obtained 0.1482 and 0.07199 respectively.

Of course we benefit in these examples from the fact that the true distribution is really log-concave. On the other hand, many distributions with a non-decreasing hazard rate are log-concave (see \citealp{Duembgen_Rufibach_2009}) and in the case of slight misspecification of the model, at least for exact data, the log-concave density estimator is still consistent for a close approximation of the true density (see \citealp{Duembgen_etal_2011}).

\begin{figure}[h]
\centering
\includegraphics[width=0.32\textwidth]{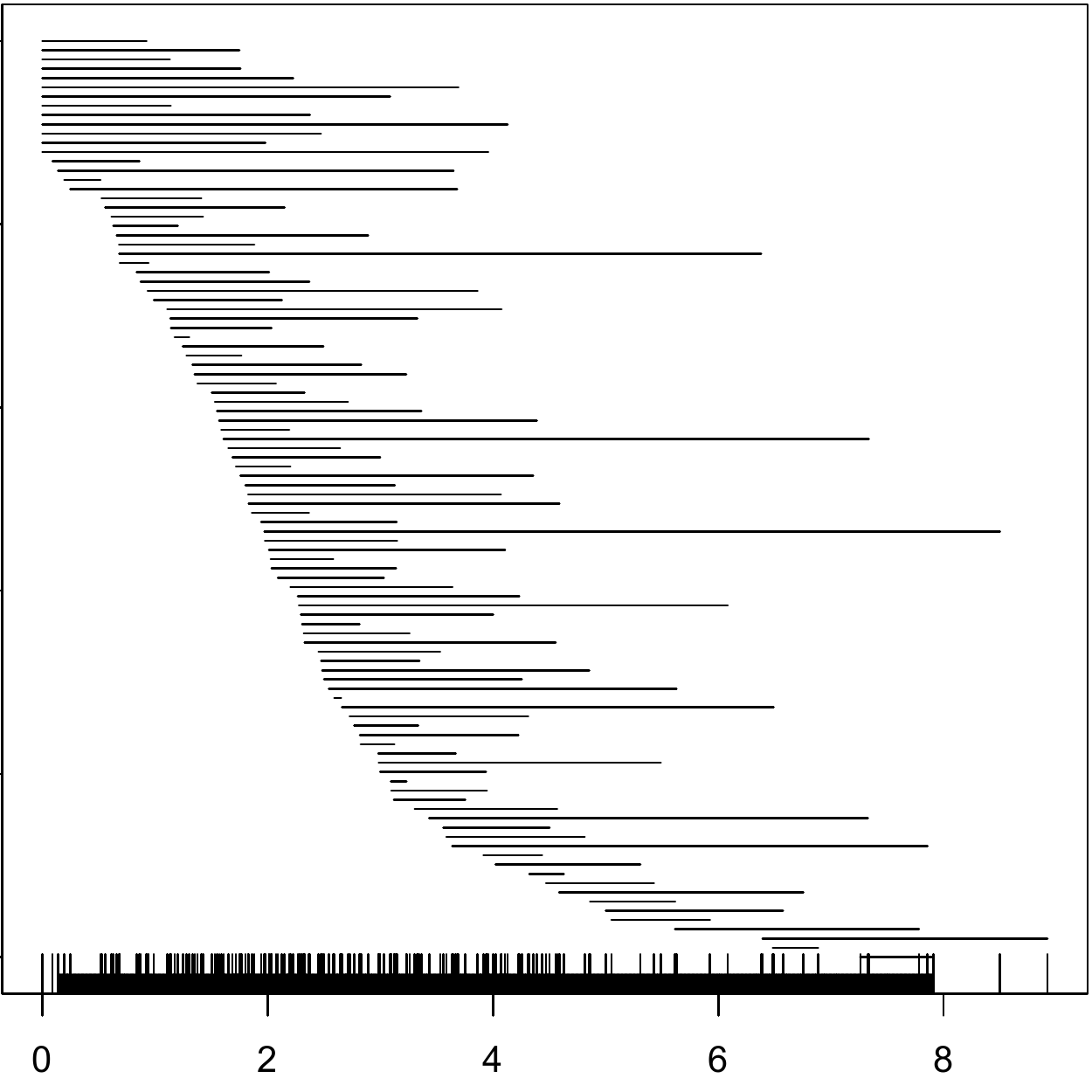} \; \includegraphics[width=0.32\textwidth]{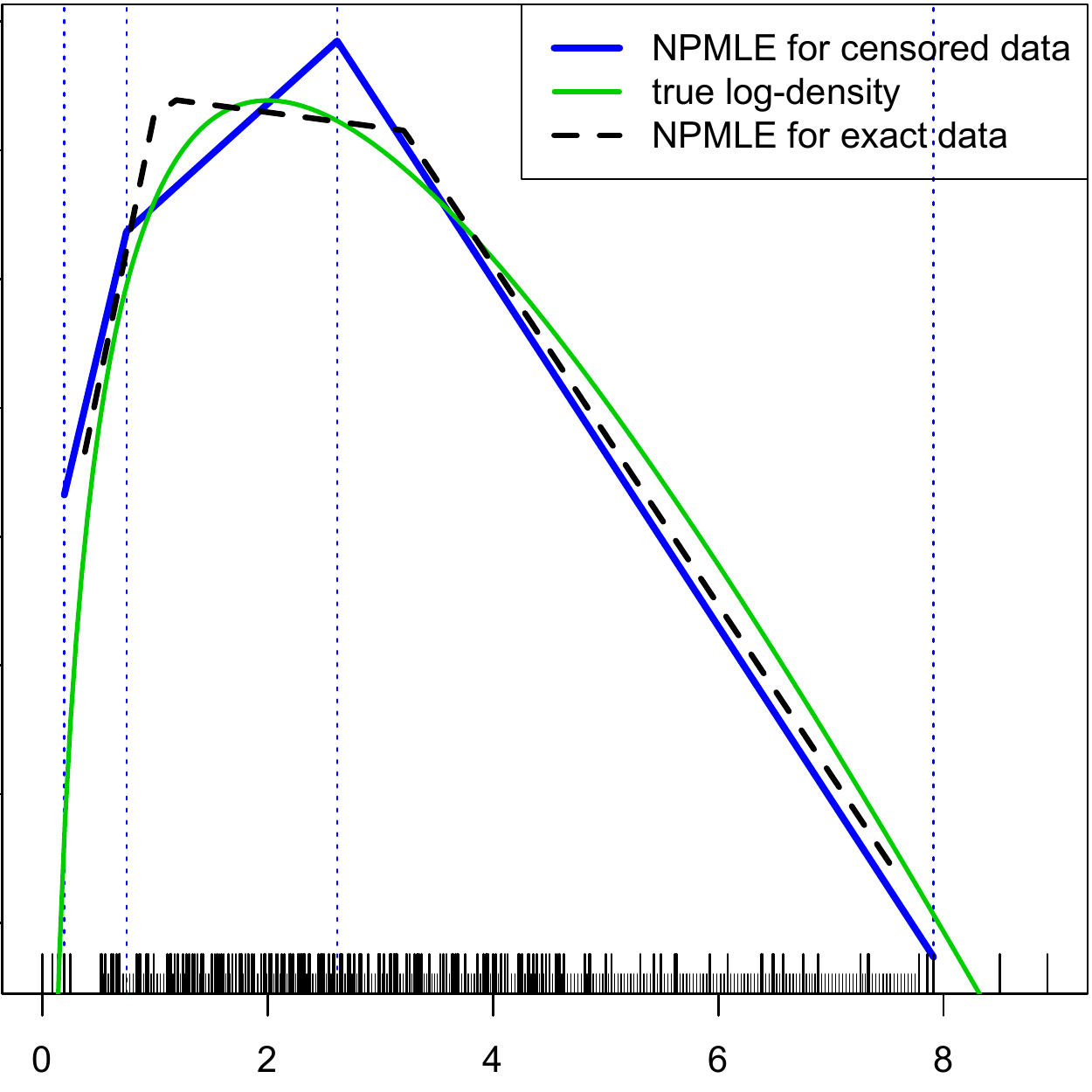} \; \includegraphics[width=0.32\textwidth]{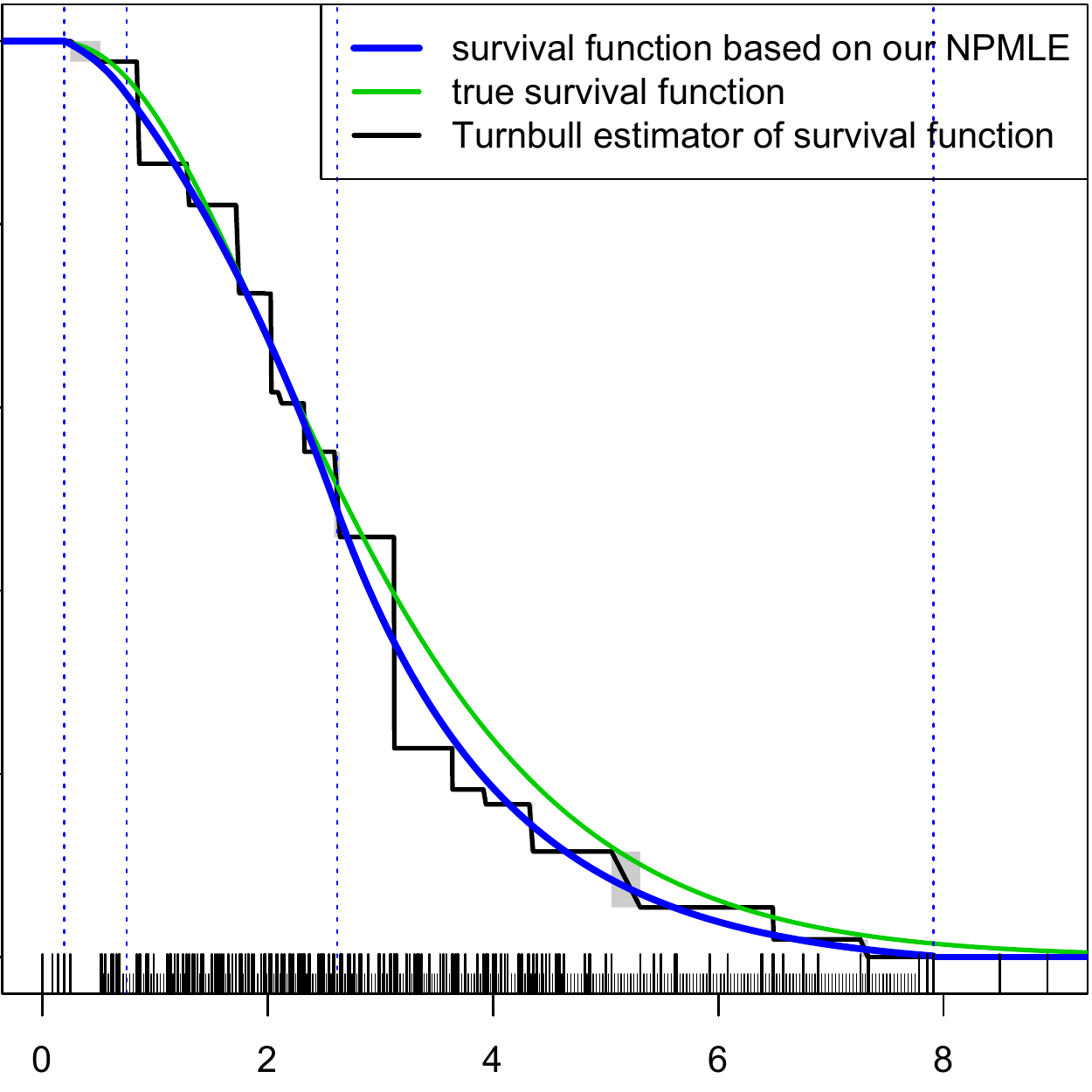}
\caption{Simulated interval-censored data and a comparison of estimators. The longer tick marks above the $x$-axis give the interval endpoints $\tau_1, \ldots, \tau_m$, the shorter tick marks the grid points $t_1,\ldots,t_N$. The faint vertical lines indicate the positions of the knots (changes of slope) of $\hat{\phi}$.}
\label{fig:simulated}
\end{figure}

\paragraph{Real data.}

We estimate the survival curve for the data from \cite{Edmunson_etal_1979}, which is available in the dataset {\tt ovarian} in the R package {\tt survival} \citep{Therneau_2013}. The survival times in days of 26 patients with advanced ovarian carcinoma were recorded along with certain covariate information, which we ignore here. Twelve observations are uncensored and the rest is right-censored.

The data is depicted in the left panel of Figure~\ref{fig:real}, where a dot represents an exact observation for a patient at a certain time. The right panel shows the survival function based on our estimator $(\hat{\phi},\hat{q})$ together with the celebrated Kaplan--Meier estimator, which is just the special case of the Turnbull estimator for right-censored data. The cure parameter is estimated at $\hat{q} \approx 0.4944$, which is just slightly below the final level of approximately $0.4967$ of the Kaplan--Meier estimator. While it becomes clear from other examples that this is a real rather than a numerical difference, this difference in the cure parameters typically tends to be small.

\begin{figure}[h]
\centering
\includegraphics[width=0.42\textwidth]{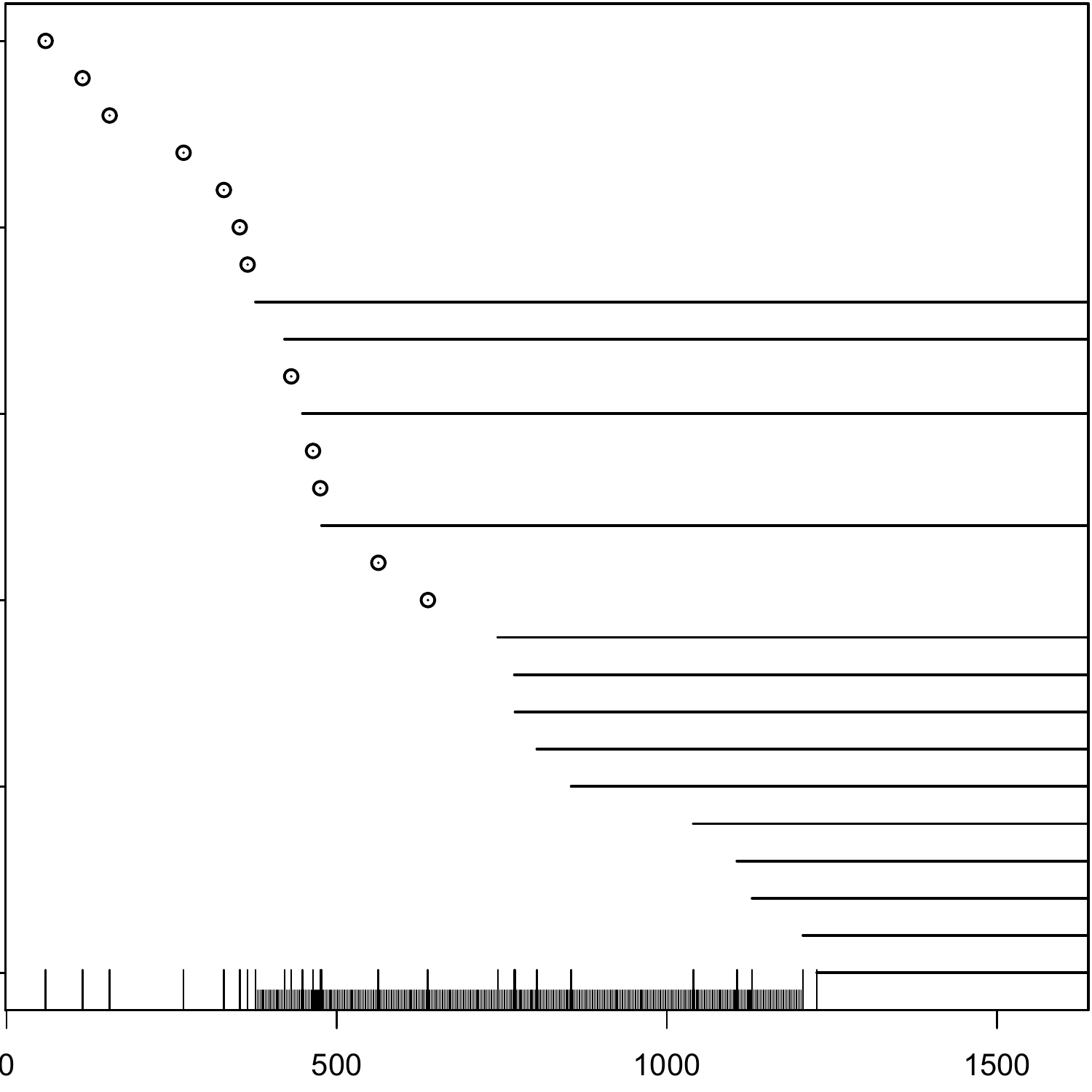} \hspace*{10mm} \includegraphics[width=0.42\textwidth]{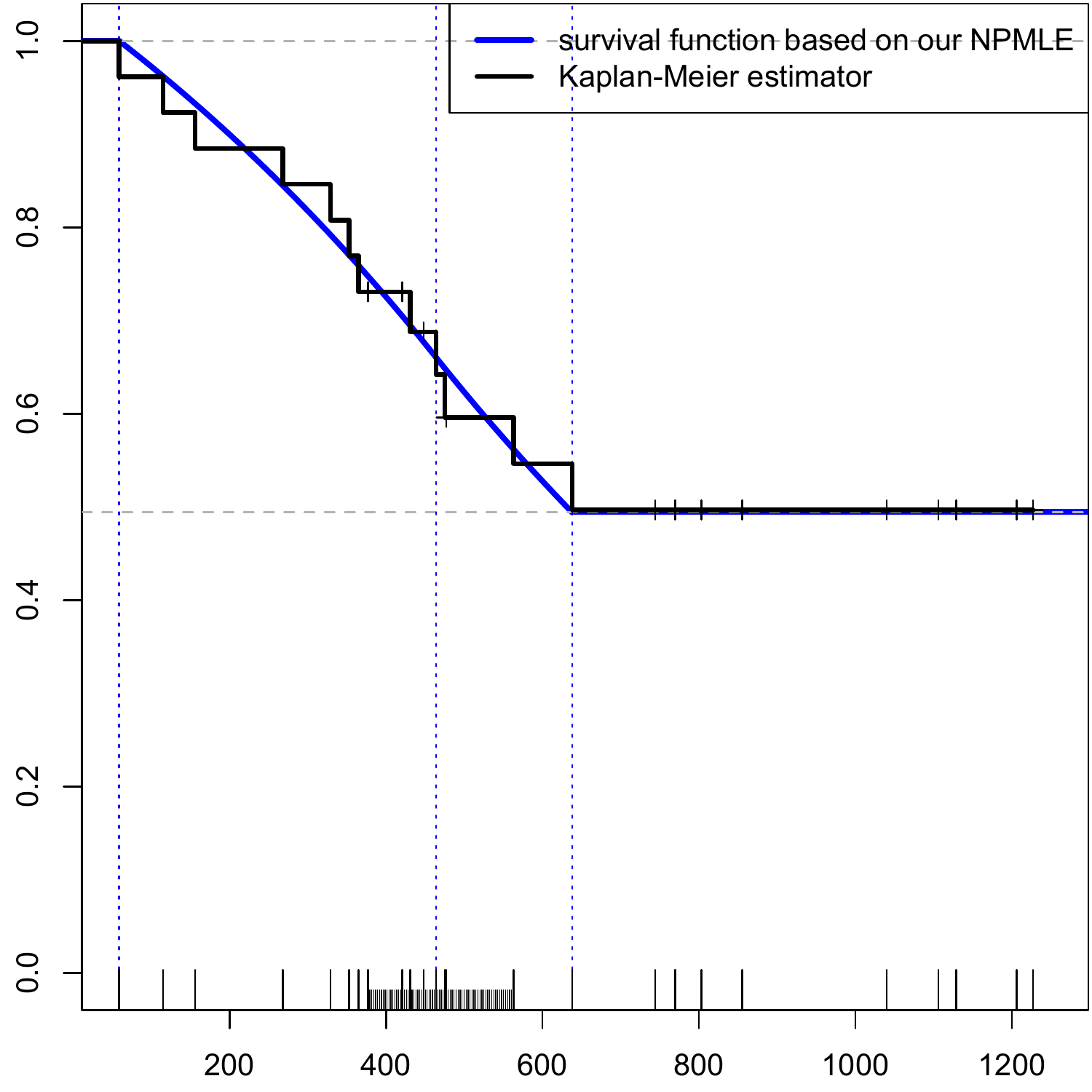}
\caption{The ovarian cancer data and a comparison of estimators. Note that the set of grid points $t_1,\ldots,t_N$ was adapted according to the rules of Theorem~\ref{thm:shape} after each of the domain reduction steps proposed in Section~\ref{sec:algorithms}. Consequently the final grid (right panel) is concentrated on a much smaller part of the time axis than the original one (left panel).}
\label{fig:real}
\end{figure}

\section{Proofs and technical results}
\label{sec:proofs}


An essential ingredient for the proof of Theorem~\ref{thm:existence} are the following inequalities:

\begin{Lemma}
\label{lem:inequalities}
Let $\phi : \R \to [-\infty,\infty)$ be a concave function such that $\int_{\R} e^{\phi(t)} \, dt \le 1$, and let $x \in \dom(\phi)$. Then for any $y \in \R$ with $|x - y| \ge e^{-\phi(x)}$,
\[
	\phi(y) \ \le \ \phi(x) + 1 - e^{\phi(x)} |y - x| .
\]
Moreover, for any $\delta > 0$ and any interval $B \subset [x + \delta, \infty)$ or $B \subset (-\infty,x-\delta]$,
\[
	\int_B e^{\phi(t)} \, dt
	\ \le \ \exp(- e^{\phi(x)} \delta) .
\]
\end{Lemma}

\begin{proof}[\bf Proof of Lemma~\ref{lem:inequalities}]
Note first that by concavity of $\phi$, convexity of the exponential function and Jensen's inequality,
\begin{align}
	\nonumber
	\int_{\min(x,y)}^{\max(x,y)} e^{\phi(t)} \, dt \
	&= \ |y - x| \int_0^1 \exp \bigl( \phi((1 - u)x + uy) \bigr) \, du \\
	\label{eq:helper1a}
	&\ge \ |y - x| \int_0^1 \exp \bigl( (1 - u) \phi(x) + u \phi(y) \bigr) \, du \\
	\label{eq:helper1b}
	&\ge \ |y - x| \exp \bigl( \phi(x)/2 + \phi(y)/2 \bigr) .
\end{align}
Since the left hand side is less than or equal to one, and since the right hand side of \eqref{eq:helper1b} may be written as $|y - x| e^{\phi(x)} \exp \bigl( \phi(y) - \phi(x) \bigr)^{1/2}$, it follows from $|y - x| \ge e^{-\phi(x)}$ that $\gamma := \phi(x) - \phi(y) \ge 0$. But then the right hand side of \eqref{eq:helper1a} equals
\[
	|y - x| e^{\phi(x)} \int_0^1 \exp(- \gamma u) \, du
	\ = \ |y - x| e^{\phi(x)} (1 - e^{-\gamma})/\gamma
\]
with $(1 - e^{-0})/0 := 1$. Since $(1 - e^{-\gamma})/\gamma \ge 1/(1 + \gamma)$ for arbitrary $\gamma > 0$, we may conclude that $1 + \gamma \ge e^{\phi(x)} |y - x|$, which is equivalent to $\phi(y) \le \phi(x) + 1 - |y - x| e^{\phi(x)}$.

As for the second part, let $B \subset [x + \delta, \infty)$. If we define $\phi_o(t) := \phi(x) - e^{\phi(x)} (t - x)$, then $\int_x^\infty e^{\phi_o(t)} \, dt = 1 \ge \int_x^\infty e^{\phi(t)} \, dt$, and by concavity of $\phi$ there exists a number $y_o \in [x,\infty)$ such that $\phi \ge \phi_o$ on $[0,y_o)$ and $\phi \le \phi_o$ on $(y_o,\infty)$. In case of $x + \delta \ge y_o$,
\[
	\int_B e^{\phi(t)} \, dt
	\ \le \ \int_{x + \delta}^\infty e^{\phi_o(t)} \, dt
	\ = \ \exp(- e^{\phi(x)}\delta) .
\]
In case of $x + \delta < y_o$,
\[
	\int_B e^{\phi(t)} \, dt
	\ \le \ 1 - \int_{x}^{x + \delta} e^{\phi(t)} \, dt
	\ \le \ 1 - \int_{x}^{x + \delta} e^{\phi_o(t)} \, dt
	\ = \ \exp(- e^{\phi(x)}\delta) .
\]
The case $B \subset (-\infty,x-\delta]$ may be treated analogously.
\end{proof}

Another important ingredient for proving Theorem~\ref{thm:existence} is a slight modification of Lemma~4.2 of \citet{Duembgen_etal_2011} which we state without proof:

\begin{Lemma}
\label{lem:compactness}
Let $\overline{\phi}$ and $\phi_1, \phi_2, \phi_3, \ldots$ be concave and upper semicontinuous functions from $\R$ into $[-\infty,\infty)$ such that $\phi_k \le \overline{\phi}$ for all $k$. Further suppose that for a compact interval $[a,b] \subset \R$,
\[
	\liminf_{k \to \infty} \max_{x \in [a,b]} \phi_k(x) \ > \ - \infty .
\]
Then there exists a concave and upper semicontinuous function $\phi : \R \to [-\infty,\infty)$ and a subsequence $(\phi_{k(j)})_j$ of $(\phi_k)_k$ such that
\begin{align*}
	\limsup_{j \to \infty, \, x \to x_o} \phi_{k(j)}(x) \
	&\le \ \phi(x_o) \quad\text{for any} \ x_o \in \R , \\
	\lim_{j \to \infty, \, x \to x_o} \phi_{k(j)}(x) \
	&= \ \phi(x_o) \quad\text{for any} \ x_o \in \mathrm{interior}(\dom(\phi)) .
\end{align*}
Moreover, $\dom(\phi) \cap [a,b] \ne \emptyset$.
\end{Lemma}

\begin{proof}[\bf Proof of Theorem~\ref{thm:existence}]
We first consider $\ell(\cdot,0)$. According to Lemmas~\ref{lem:exotic1} and \ref{lem:exotic2} it suffices to consider data sets such that
\[
	\bigcap_{i=1}^n [L_i,R_i] \ = \ \emptyset .
\]
In other words, there exist indices $i(1), i(2) \in \{1,2,\ldots,n\}$ such that
\[
	a' := L_{i(1)} \le R_{i(1)} =: a \ < \ b := L_{i(2)} \le R_{i(2)} =: b' .
\]
Now let $(\phi_k)_k$ be a sequence in $\Phi(0)$ such that $- \infty < \ell(\phi_k,0) \to \sup_{\phi \in \Phi(0)} \ell(\phi,0)$ as $k \to \infty$. This implies that $M_k := \max_{x \in \R} \phi_k(x)$ is bounded. For if $M_k \ge \log((b - a)/2)$ and $x_k$ is a maximizer of $\phi_k$, then in case of $x_k \le (a + b)/2$,
\begin{align*}
	n \ell(\phi_k,0) \ = \
	&\sum_{i \ne i(2)}
		\Bigl( 1_{[L_i = R_i]} \phi_k(X_i)
			+ 1_{[L_i < R_i]} \log \Bigl( \int_{L_i}^{R_i} e^{\phi_k(t)} \, dt \Bigr)
				\Bigr) \\
	& + \ 1_{[b = b']} \phi_k(b)
		+ 1_{[b < b']} \log \Bigl( \int_b^{b'} e^{\phi_k(t)} \, dt \Bigr) \\
	\le \
	&(n - 1) M_k^+ \\
	& + \ 1_{[b = b']} \phi_k(b)
		+ 1_{[b < b']} \log \Bigl( \int_b^\infty e^{\phi_k(t)} \, dt \Bigr) \\
	\le \
	&(n - 1) M_k^+ + (M_k + 1)^+ - (b - x_k) e^{M_k} \\
	\le \
	&n (M_k + 1)^+ - (b - a) e^{M_k}/2
\end{align*}
by Lemma~\ref{lem:inequalities}. Analogous arguments may be applied in case of $x_k \ge (a + b)/2$. This yields both times the inequality
\[
	n \ell(\phi_k,0) \ \le \ n (M_k + 1)^+ - (b - a) e^{M_k}/2 ,
\]
and the right hand side tends to $-\infty$ as $M_k \to \infty$.

Now let $M$ be an upper bound for all maxima $M_k$. Then it follows from
\begin{align*}
	n \ell(\phi_k,0)
	\ \le \ 
	&(n - 1) M^+ \\
	& + \ 1_{[a' = a]} \phi_k(a)
		+ 1_{[a' < a]} \bigl( \max\{\phi_k(a'),\phi_k(a)\}
			+ \log(a - a') \bigr) \\
	\le \
	&(n - 1) M^+ + \max\{\phi_k(a'),\phi_k(a)\}
		+ \log(a - a')^+
\end{align*}
that $\max\{\phi_k(a'),\phi_k(a)\}$ is bounded away from $-\infty$, say,
\[
	\max\{\phi_k(a'),\phi_k(a)\} \ \ge \ m > - \infty
\]
for all $k$. But then it follows from Lemma~\ref{lem:inequalities} that
\[
	\phi_k(x)
	\ \le \ \overline{\phi}(x)
	\ := \ M - e^{m} (x - a - e^{-m})^+ - e^{m} (a' - x - e^{-m})^+
\]
for all $k$ and any $x \in \R$.

Hence we may apply Lemma~\ref{lem:compactness} to conclude that after replacing $(\phi_k)_k$ with a subsequence, if necessary, there exists a concave and upper semicontinuous function $\phi : \R \to [-\infty,\infty)$ such that
\begin{align*}
	\limsup_{k \to \infty, x' \to x} \phi_k(x') \
	&\le \ \phi(x) \quad\text{for any} \ x \in \R , \\
	\lim_{k \to \infty, x' \to x} \phi_k(x') \
	&= \ \phi(x) \quad\text{for any} \ x \in \mathrm{interior}(\dom(\phi)) .
\end{align*}
In particular, $\lim_{k \to \infty} \phi_k(x) = \phi(x)$ for all but at most two points $x \in \R$. By dominated convergence, $\int_{\R} e^{\phi(x)} \, dx = 1$ and $\int |e^{\phi_k(x)} - e^{\phi(x)}| \, dx \to 0$ as $k \to \infty$. Consequently,
\[
	\lim_{k \to \infty} \ell(\phi_k,0)
	\ \le \ \ell(\phi,0) ,
\]
i.e.\ $\phi$ is a maximizer of $\ell(\cdot,0)$ over $\Phi(0)$.

Now we consider maximization of $\ell(\cdot)$ over $\Theta$. Without loss of generality we assume that $\infty \in \{R_1,R_2,\ldots,R_n\} \ne \{\infty\}$. For if $R_i = \infty$ for all $i$, then we are in the situation of Lemma~\ref{lem:exotic1} with $\mu'' = \infty$. If $R_i < \infty$ for all $i$, then $\ell(\phi - \log(1 - q), 0) = \ell(\phi,q) - \log(1 - q)$ for arbitrary $(\phi,q) \in \Theta$, so we are again maximizing $\ell(\cdot,0)$.

Let $\bigl( (\phi_k,q_k) \bigr)_k$ be a sequence in $\Theta$ such that $-\infty < \ell(\phi_k,q_k) \to \sup_{(\phi,q) \in \Theta} \ell(\phi,q)$. In addition we may and do assume that $\lim_{k \to \infty} q_k = q_o \in [0,1]$. Again let $x_k$ be a maximizer of $\phi_k$ and set $M_k := \phi_k(x_k)$. We first show that $(M_k)_k$ may be assumed to be bounded. Note that by Lemma~\ref{lem:inequalities},
\[
	n \ell(\phi_k,q_k)
	\ \le \ \sum_{i=1}^n 1_{[R_i < \infty]} \bigl( (M_k+1)^+ - e^{M_k} \rho_{ik} \bigr)
\]
with
\[
	\rho_{ik} \ := \ \min_{y \in [L_i,R_i]} |y - x_k|
	\ = \ \begin{cases}
		0 & \text{if} \ x_k \in [L_i,R_i] , \\
		\min\{|L_i - x_k|, |R_i - x_k|\} & \text{if} \ x_k \not\in [L_i,R_i] .
	\end{cases}
\]
If $M_k \to \infty$, then $\rho_{ik} \to 0$ for all $i$ with $R_i < \infty$. This implies that
\[
	\bigcap_{i \,:\, R_i < \infty} [L_i,R_i] \ = \ [\mu',\mu'']
\]
for certain real numbers $\mu',\mu''$ with $\mu' \le \liminf_{k\to\infty} x_k \le \limsup_{k\to\infty} x_k \le \mu''$.

Suppose first that $\tilde{X}_{i_o} = \{X_{i_o}\}$ for some $i_o \in \{1,2,\ldots,n\}$. Then $\mu' = \mu'' = X_{i_o}$, and it follows from Lemma~\ref{lem:exotic3} that no maximizer of $\ell(\cdot)$ exists.

If there are no uncensored observations, then for each $i \in \{1,2,\ldots,n\}$ either $L_i < \mu''$ or $\mu'' \le L_i < R_i \le \infty$. Now we define
\begin{align*}
	a \
	&:= \ \max \{L_i : L_i < \mu''\}
		\ \in \ (-\infty, \mu'') , \\
	b \
	&:= \ \min \bigl\{ \{\mu'' + 1\} \cup \{L_i : L_i > \mu''\} \bigr)
		\ \in \ (\mu'',\infty) ,
\end{align*}
and
\[
	p_{k\ell}
	\ := \ \int_{-\infty}^{\mu''} e^{\phi(t)} \, dt ,
	\quad
	p_{kr}
	\ := \ \int_{\mu''}^{\infty} e^{\phi(t)} \, dt ,
\]
so $p_{k\ell} + p_{kr} + q_k = 1$. Then
\begin{align*}
	n \ell(\phi_k,q_k) \
	&\le \ \#\{i : L_i < R_i = \mu''\} \log(p_{k\ell}) \\
	& + \ \#\{i : L_i < \mu'' < R_i < \infty\} \log(1 - q_k) \\
	& + \ \#\{i : L_i = \mu'' < R_i < \infty\} \log(p_{kr}) \\
	& + \ \#\{i : L_i = \mu'' < R_i = \infty\} \log(p_{kr} + q_k) \\
	& + \ \#\{i : L_i > \mu''\} \log \bigl( \exp(- e^{M_k} (b - x_k)) + q_k \bigr) .
\end{align*}
This implies that $\liminf_{k \to \infty} q_k > 0$ if $L_i > \mu''$ for some $i$, because $\lim_{k \to \infty} e^{M_k} (b - x_k) = \infty$. Thus we may conclude that
\begin{align*}
	n \ell(\phi_k,q_k) \
	&\le \ \#\{i : L_i < R_i = \mu''\} \log(p_{k\ell}) \\
	& + \ \#\{i : L_i < \mu'' < R_i < \infty\} \log(1 - q_k) \\
	& + \ \#\{i : L_i = \mu'' < R_i < \infty\} \log(p_{kr}) \\
	& + \ \#\{i : L_i = \mu'' < R_i = \infty\} \log(p_{kr} + q_k) \\
	& + \ \#\{i : L_i > \mu''\} \log(q_k) + o(1) \\
	& = \ n \ell(\tilde{\phi}_k,q_k) + o(1) ,
\end{align*}
where
\[
	\tilde{\phi}_k(t) \ := \ \begin{cases}
		- \log(\min\{\mu''-a,b-\mu''\}) - \gamma_{k\ell} (\mu'' - t)
		& \text{for} \ t \in [a,\mu''] , \\
		- \log(\min\{\mu''-a,b-\mu''\}) - \gamma_{kr} (t - \mu'')
		& \text{for} \ t \in [\mu'',b] , \\
		- \infty
		& \text{for} \ t \in \R \setminus [a,b] ,
	\end{cases}
\]
and $\gamma_{k\ell}, \gamma_{kr} \ge 0$ are chosen such that $\int_{a}^{\mu''} e^{\tilde{\phi}_k(t)} \, dt = p_{k\ell}$ and $\int_{\mu''}^b e^{\tilde{\phi}_k(t)} \, dt = p_{kr}$.

The previous considerations show that, after replacing $(\phi_k)_k$ with a surrogate sequence if necessary, we may assume that $\phi_k \le M$ for all $k$ and some real constant $M$. Next we show that the limit $q_o$ of $(q_k)_k$ is strictly smaller than one. Note that
\[
	n \ell(\phi_k,q_k)
	\ \le \ \# \{i : L_i = R_i\} M + \# \{i : L_i < R_i < \infty\} \log(1 - q_k) ,
\]
so $q_o = 1$ would imply that each observation has to be uncensored or of the form $(L_i,\infty]$. If all uncensored observations would be identical, we could conclude from Lemma~\ref{lem:exotic3} that there exists no maximizer of $\ell(\dot)$. If $a := L_{i(1)} = R_{i(1)} < b := L_{i(2)} = R_{i(2)}$ for certain indices $i(1), i(2) \in \{1,2,\ldots,n\}$, then
\[
	n \ell(\phi_k,q_k)
	\ \le \ (n-1) M + \min\{\phi_k(a),\phi_k(b)\} .
\]
Hence $\min\{\phi_k(a),\phi_k(b)\} \ge m$ for all $k$ and a certain number $m > -\infty$. But then $q_k \le 1 - \int_a^b e^{\phi_k(t)} \, dt \le 1 - (b - a) e^m$, a contradiction to $\lim_{k \to \infty} q_k = 1$.

Thus we may assume that $\phi_k \le M$ for all $k$ and $\lim_{k \to \infty} q_k = q_o \in [0,1)$. Let $[a,b] := [L_{i_o},R_{i_o}]$ for some $i_o$ with $R_{i_o} < \infty$. Then
\[
	n \ell(\phi_k,q_k)
	\ \le \ (n - 1) M^+ + \max_{x \in [a,b]} \phi_k(x) + 1_{[a < b]} \log(b - a) .
\]
Consequently, $\max_{x \in [a,b]} \phi_k(x) \ge m$ for all $k$ and some real number $m$. Hence Lemma~\ref{lem:inequalities} implies that
\[
	\phi_k(x)
	\ \le \ \overline{\phi}(x)
	\ := \ M - e^m (x - b - e^{-m})^+ - e^m (a - x - e^{-m})^+ .
\]
Again we may apply Lemma~\ref{lem:compactness} and dominated convergence to conclude that there exists a function $\phi \in \Phi(q_o)$ such that $\limsup_{k \to \infty} \ell(\phi_k,q_k) \le \ell(\phi,q_o)$.
\end{proof}


\begin{proof}[\bf Proof of Lemma~\ref{lem:exotic1}]
Note first that all observations satisfy $L_i \le \mu' < \mu'' \le R_i$. Hence
\[
	\ell(\phi,q)
	\ = \ \frac{1}{n} \sum_{i=1}^n \log P_{\phi,q}((L_i,R_i])
	\ \le \ 0
\]
with equality if, and only if, $P_{\phi,q}((L_i,R_i]) = 1$ for $1 \le i \le n$. But this is easily shown to be equivalent to $P_{\phi,q}((\mu',\mu'']) = 1$.
\end{proof}

\begin{proof}[\bf Proof of Lemma~\ref{lem:exotic2}]
Suppose first that $\{\mu\} \subset [L_i,R_i]$ for all indices $i$ with at least one equality. For $\eps > 0$,
\[
	\phi_\eps(x) \ := \ - \log(2\eps) - |x - \mu|/\eps
\]
defines a log-density in $\Phi(0)$ such that
\begin{align*}
	\lim_{\eps \downarrow 0} \, \phi_\eps(\mu)
	\ &= \ \infty , \\
	\lim_{\eps \downarrow 0} \int_{L_i}^{R_i} e_{}^{\phi_\eps(x)} \, dx
	\ &= \ \begin{cases}
			1/2 & \text{if} \ L_i < \mu = R_i \ \text{or} \ L_i = \mu < R_i , \\
			1   & \text{if} \ L_i < \mu < R_i .
		\end{cases}
\end{align*}
Hence $\ell(\phi_\eps) \to \infty$ as $\eps \downarrow 0$ which implies the first assertion.

If $\bigcap_{i=1}^n [L_i,R_i] = \{\mu\}$ but $[L_i,R_i] \ne \{\mu\}$ for all indices $i$, then
\[
	\ell(\phi,q)
	\ = \ \frac{1}{n} \sum_{i=1}^n \log P((L_i,R_i])
\]
with $P := P_{\phi,q}$. But
\[
	P((L_i,R_i])
	\ \le \ \begin{cases}
		P((-\infty,\mu]) & \text{if} \ L_i < \mu = R_i \\
		1 & \text{if} \ L_i < \mu < R_i \\
		P((\mu,\infty]) & \text{if} \ L_i = \mu < R_i
	\end{cases}
\]
with equality if, and only if, $P((-\infty,a]) = P((b,\infty]) = 0$. Thus
\[
	\ell(\phi,q)
	\ \le \ \frac{n_\ell}{n} \log P((-\infty,\mu])
		+ \frac{n_r}{n} \log P((\mu,\infty])
\]
with equality if, and only if, $P((-\infty,\mu]) = P((a,\mu])$ and $P((\mu,\infty]) = P((\mu,b])$. Writing $x := P((\mu,\infty]) \in [0,1]$, we end up with the upper bound
\[
	\ell(\phi,q) \ \le \ \frac{n_\ell}{n} \log(1 - x) + \frac{n_r}{n} \log(x) .
\]
Finally, this bound becomes maximal if, and only if, $x = n_r/(n_\ell + n_r)$.
\end{proof}

\begin{proof}[\bf Proof of Lemma~\ref{lem:exotic3}]
We fix an arbitrary value $q \in (0,1)$ for $P(\{\infty\})$. Then
\[
	\phi_\eps(x) \ := \ \log(1 - q) - \log(2\eps) - |x - \mu|/\eps
\]
defines a function in $\Phi(q)$ such that $\ell(\phi_\eps,q) \to \infty$ as $\eps \to 0$, because $\lim_{\eps \to 0} \phi_{\eps}(\mu) = \infty$ while $\liminf_{\eps \to 0} P_{\phi_\eps,q}((a,b]) \ge \min\{(1-q)/2,q\}$, whenever $a < b$ and $\mu \in [a,b]$ or $b = \infty$.
\end{proof}

\begin{proof}[\bf Proof of Lemma~\ref{lem:basic}]
Let $(\phi,q) \in \Theta$ such that $\ell(\phi,q) > -\infty$ and $\delta := P_{\phi,q} \bigl( (-\infty,\infty] \setminus [a,b] \bigr) > 0$. If $\delta < 1$, then
\begin{align*}
	\tilde{\phi}(x) \
		&:= \ \begin{cases}
			\phi(x) - \log(1 - \delta) & \text{for} \ x \in [a,b] \\
			-\infty & \text{for} \ x \not\in [a,b]
			\end{cases} \\
	\tilde{q} \
		&:= \ q / (1 - \delta)
\end{align*}
defines a new pair $(\tilde{\phi},\tilde{q}) \in \Theta$ such that $\ell(\tilde{\phi},\tilde{q}) = \ell(\phi,q) - \log(1 - \delta)$ and $\dom(\tilde{\phi}) \subset [a,b]$.

If $\delta = 1$, then we would be in one of the following two situations:

\noindent
Situation~1: $a \in \dom(\phi) \subset (-\infty,a]$, and all observations are equal to $\{a\}$ or contain $\infty$. If all observations are equal to $\{a\}$, then Lemma~\ref{lem:exotic2} would apply and exclude the existence of $\hat{\phi}_0$ or $(\hat{\phi},\hat{q})$. If at least one observations contains $\infty$, then $\ell(\phi,0) = -\infty$, and Lemma~\ref{lem:exotic3} would exclude the existence of $(\hat{\phi},\hat{q})$.

\noindent
Situation~2: $b \in \dom(\phi) \subset [b,\infty)$, and all observations are equal to $\{b\}$. Here Lemma~\ref{lem:exotic2} would exclude the existence of $\hat{\phi}_0$ or $(\hat{\phi},\hat{q})$.
\end{proof}

Our proof of Theorem~\ref{thm:shape} is based on the following two results:

\begin{Lemma}
\label{lem:linearize.phi.1}
Let $a < b < c$ be real numbers and $\phi : [a,c] \to \R$ continuous and concave. Then there exist real numbers
\[
	\gamma_\ell \ \in \ \Bigl[ \frac{\phi(b) - \phi(a)}{b - a}, \phi'(a\,+) \Bigr]
	\quad\text{and}\quad
	\gamma_r \ \in \ \Bigl[ \phi'(c\,-), \frac{\phi(c) - \phi(b)}{c - b} \Bigr]
\]
such that
\[
	\tilde{\phi}(t)
	\ := \
	\min \bigl\{ \phi(a) + \gamma_\ell (t - a),
	             \phi(c) + \gamma_r (t - c) \bigr\}
\]
satisfies
\[
	\int_a^b e^{\tilde{\phi}(t)} \, dt = \int_a^b e^{\phi(t)} \, dt
	\quad\text{and}\quad
	\int_b^c e^{\tilde{\phi}(t)} \, dt = \int_b^c e^{\phi(t)} \, dt .
\]
\end{Lemma}

\begin{Lemma}
\label{lem:linearize.phi.2}
Let $-\infty < a < c \le \infty$ and $B := [a,c] \cap \R$. Further let $\phi : B \to [-\infty,\infty)$ be concave and upper semicontinuous such that $\phi(a) > - \infty$ and $0 < \int_B e^{\phi(x)} \, dx < \infty$.

\noindent
\textbf{(i)} \ Let $\gamma$ be the unique real number such that $\tilde{\phi}(x) := \phi(a) + \gamma (x - a)$ satisfies the equation $\int_B e^{\tilde{\phi}(x)} \, dx = \int_B e^{\phi(x)} \, dx$. Then $\tilde{\phi}(a) = \phi(a)$, $\gamma \le \phi'(a \, +)$, and $\tilde{\phi}(c) \ge \phi(c)$ in case of $c < \infty$. The latter two inequalities are strict unless $\tilde{\phi} \equiv \phi$.

\noindent
\textbf{(ii)} \ Suppose that $c < \infty$ and $\phi(c) > - \infty$, $\phi'(c\,-) > -\infty$. Let $\gamma$ be the unique real number such that
\[
	\tilde{\phi}(x)
	\ := \ \min \bigl\{ \phi(a) + \gamma (x - a),
		\phi(c) + \phi'(c\,-) (x - c) \bigr\}
\]
satisfies the equation $\int_B e^{\tilde{\phi}(x)} \, dx = \int_B e^{\phi(x)} \, dx$. Then $\tilde{\phi} = \phi$ on $\{a,c\}$ and $(\phi(c) - \phi(a))/(c - a) \le \gamma \le \phi'(a\,+)$.

\noindent
\textbf{(iii)} \ The function $\tilde{\phi}$ in part (i) and (ii) satisfies
\[
	\int_b^c e^{\tilde{\phi}(x)} \, dx \ \ge \ \int_b^c e^{\phi(x)} \, dx
	\quad\text{for} \ a < b < c .
\]
\end{Lemma}

Lemmas~\ref{lem:linearize.phi.1} and \ref{lem:linearize.phi.2} are illustrated in Figures~\ref{fig:linearize.phi.1} and \ref{fig:linearize.phi.2}, respectively. In both cases one sees a strictly concave and continuous function $\phi : [a,c] \to \R$, the points $a$ and $c$ being indicated by vertical lines, and the respective surrogate functions $\tilde{\phi}$.

\begin{figure}[h]
\centering
\includegraphics[width=0.7\textwidth]{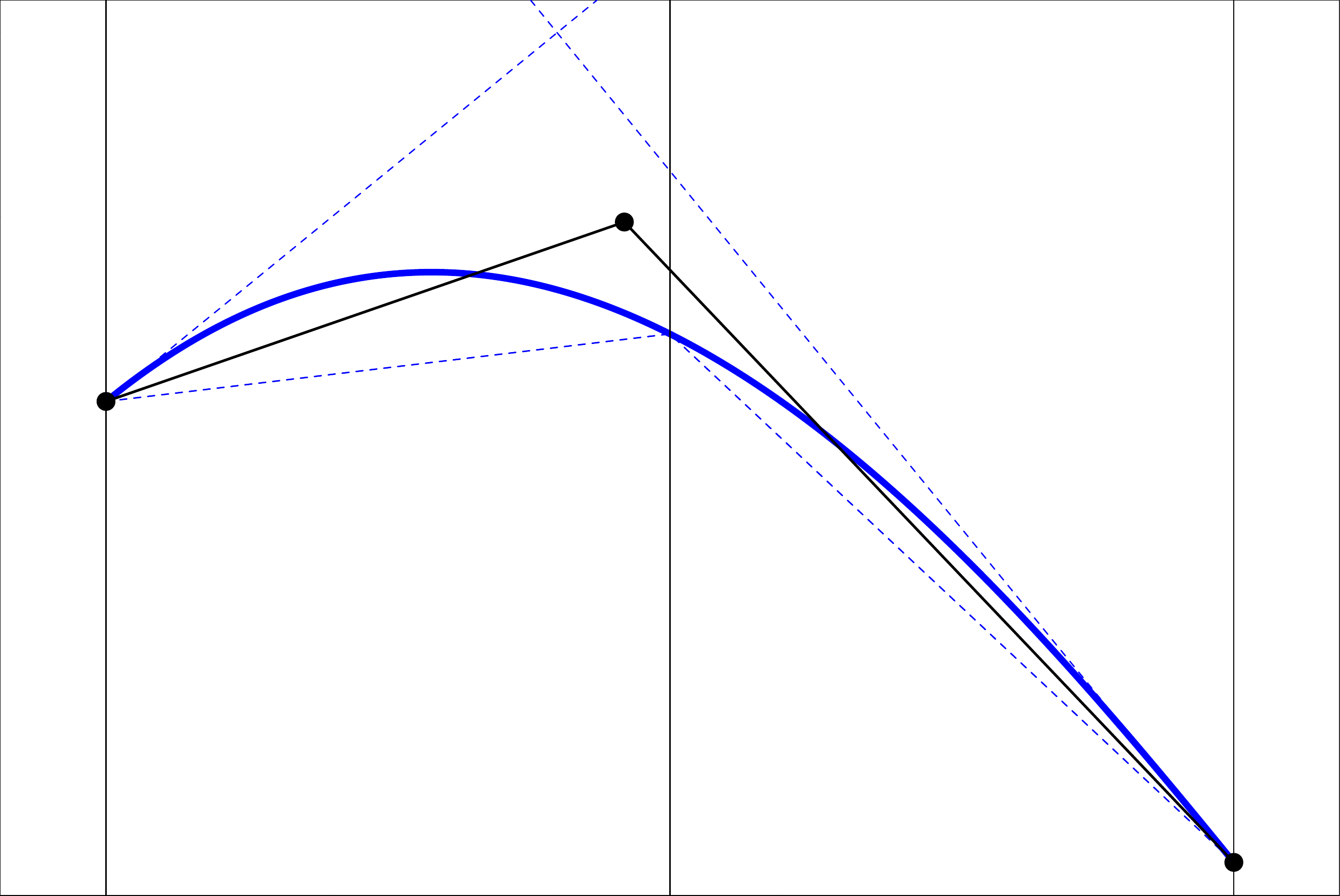}
\caption{Illustration of Lemma~\ref{lem:linearize.phi.1}.}
\label{fig:linearize.phi.1}
\end{figure}

\begin{figure}[h]
\includegraphics[width=0.49\textwidth]{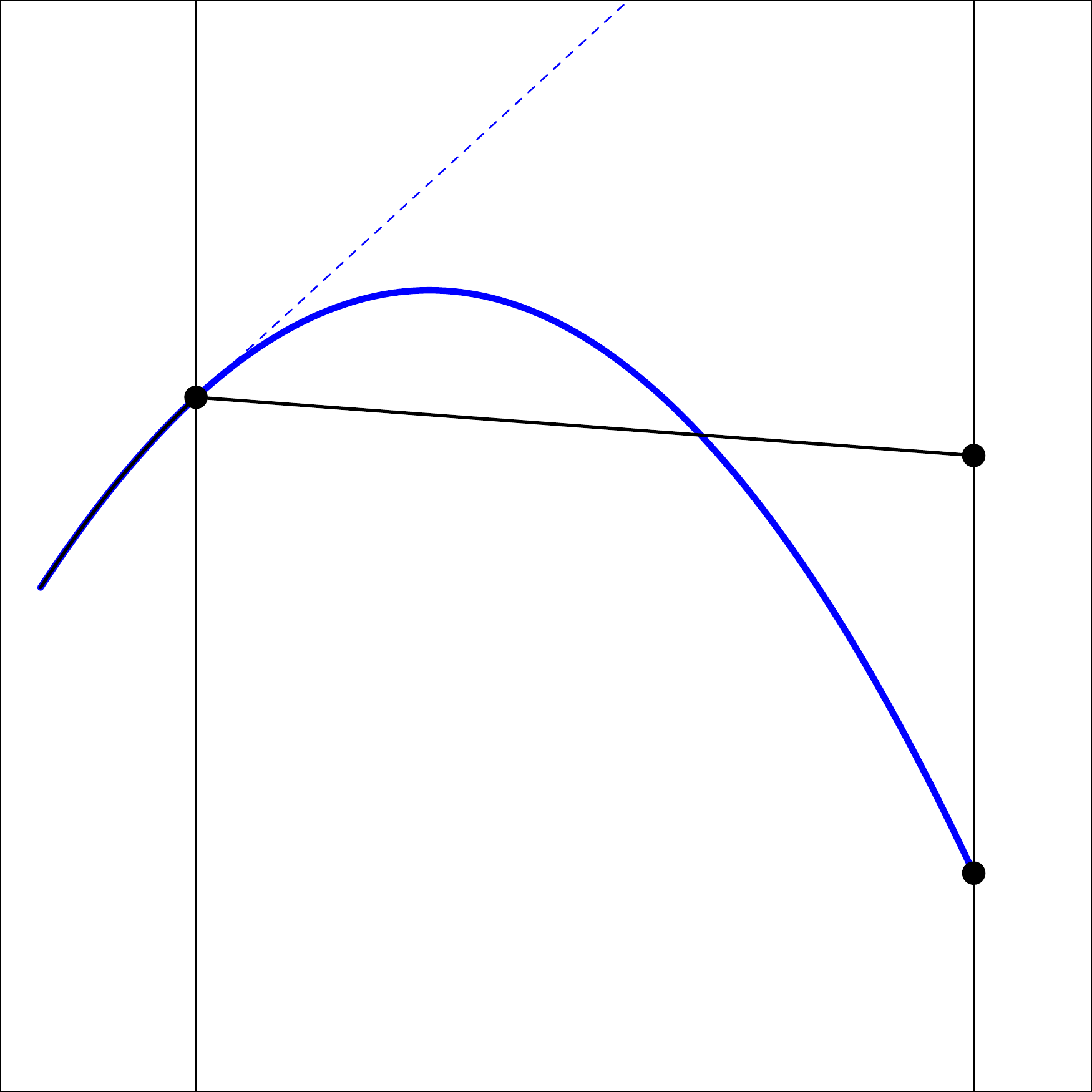}
\hfill
\includegraphics[width=0.49\textwidth]{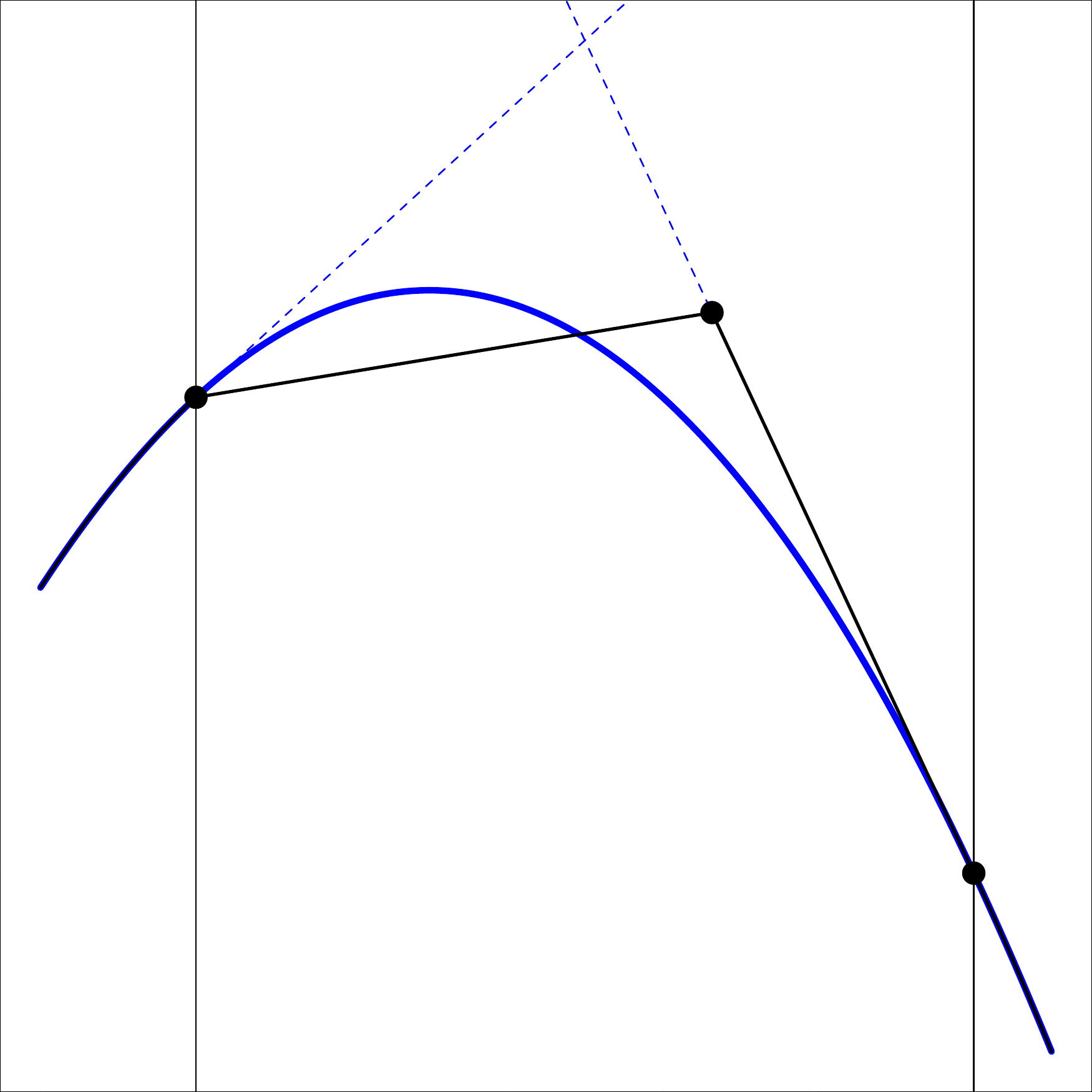}
\caption{Illustration of Lemma~\ref{lem:linearize.phi.2}, part~(i) on the left and part~(ii) on the right hand side.}
\label{fig:linearize.phi.2}
\end{figure}

\begin{proof}[\bf Proof of Lemma~\ref{lem:linearize.phi.1}]
Let
\[
	\tilde{\phi}(t)
	\ := \
	\min \bigl\{ \phi(a) + \gamma_\ell (t - a),
	             \phi(c) + \gamma_r (t - c) \bigr\}
\]
with certain constants $\gamma_\ell \ge (\phi(c) - \phi(a))/(c - a) \ge \gamma_r$ yet to be specified. This is done in two steps. First let
\[
	\gamma_\ell \ := \ \frac{y - \phi(a)}{b - a}
	\quad\text{and}\quad
	\gamma_r \ := \ \frac{\phi(c) - y}{c - b}
\]
for some real number $y \ge \phi(b)$. That means, $\tilde{\phi}$ is a triangular function connecting the points $(a,\phi(a))$, $(b,y)$ and $(c,\phi(c))$. Now we choose $y$ as large as possible such that still
\begin{align}
\label{eq:lin.phi.1.left}
	\int_a^b e^{\tilde{\phi}(t)} \, dt \
	&\le \ \int_a^b e^{\phi(t)} \, dt , \\
\label{eq:lin.phi.1.right}
	\int_b^c e^{\tilde{\phi}(t)} \, dt \
	&\le \ \int_b^c e^{\phi(t)} \, dt .
\end{align}
This means, at least one of the former two inequalities is an equality. It follows from $y \ge \phi(b)$ that $\gamma_\ell \ge (\phi(b) - \phi(a))/(b - a)$ and $\gamma_r \le (\phi(c) - \phi(b))/(c - b)$.

Now comes the second step: If \eqref{eq:lin.phi.1.left} is strict, we replace the current slope $\gamma_\ell$ by a larger value such that \eqref{eq:lin.phi.1.left} becomes an equality. Likewise, if \eqref{eq:lin.phi.1.right} is strict, we replace the current slope $\gamma_r$ by a smaller value such that \eqref{eq:lin.phi.1.left} becomes an equality. One can easily verify that $\gamma_\ell \le \phi'(a\,+)$ and $\gamma_r \ge \phi'(c\,-)$.
\end{proof}

\begin{proof}[\bf Proof of Lemma~\ref{lem:linearize.phi.2}]
Existence and uniqueness of the surrogate function $\tilde{\phi}$ follow from elementary considerations in both scenarios (i) and (ii). One can also verify easily that either $\tilde{\phi} \equiv \phi$, or $\gamma < \phi'(a\,+)$ and there exists a number $b_o \in (a,c)$ such that
\[
	\tilde{\phi} \ \begin{cases}
		\le \ \phi & \text{on} \ (a,b_o) , \\
		\ge \ \phi & \text{on} \ (b_o,c] \cap \R .
	\end{cases}
\]
The latter conditions imply the inequalities of part~(iii). For if $b \in [b_o,c)$, then the inequality $\int_b^c e^{\tilde{\phi}(x)} \, dx \ge \int_b^c e^{\phi(x)} \, dx$ is obvious. If $b \in (a,b_o]$, then
\[
	\int_b^c e^{\tilde{\phi}(x)} \, dx
	\ = \ D - \int_a^b e^{\tilde{\phi}(x)} \, dx
	\ \ge \ D - \int_a^b e^{\phi(x)} \, dx
	\ = \ \int_b^c e^{\phi(x)} \, dx ,
\]
where $D := \int_a^c e^{\tilde{\phi}(x)} \, dx = \int_a^c e^{\phi(x)} \, dx$.
\end{proof}

\begin{proof}[\bf Proof of Theorem~\ref{thm:shape}]
By means of Lemma~\ref{lem:linearize.phi.2}, applied to $\phi$ or $\phi(- \, \cdot)$, we will first construct a concave function $\tilde{\phi}$ with $\dom(\phi) \subset \dom(\tilde{\phi}) \in [\tau_1,\infty)$ such that $\tilde{\phi} \ge \phi$ on $\{\tau_1,\tau_2,\ldots,\tau_m\}$ and $\int_{\tau_j}^{\tau_{j+1}} e^{\tilde{\phi}(t)} \, dt = \int_{\tau_j}^{\tau_{j+1}} e^{\phi(t)} \, dt$ for $1 \le j \le m$. 

Precisely, let $j \in \{1,2,\ldots,m\}$. If $\dom(\phi) \cap (\tau_j,\tau_{j+1}) = \emptyset$, then we set $\tilde{\phi} := - \infty$ on $(\tau_j,\tau_{j+1})$. If $\dom(\phi) \cap (\tau_j,\tau_{j+1}) \ne \emptyset$ but $\phi(\tau_j) = \phi(\tau_{j+1}) = -\infty$, then $\dom(\phi) \subset (\tau_j,\tau_{j+1})$, and for $x \in [\tau_j,\tau_{j+1}] \cap \R$ we may define $\tilde{\phi}(x) := - \log(\tau_{j+1} - \tau_j)$ if $\tau_j < \infty$ or $\tilde{\phi}(x) = - (x - \tau_m)$ if $j = m$.

Suppose that $\tau_j \in \dom(\phi)$ but either $\tau_{j+1} \not\in \dom(\phi)$ or $\phi'(\tau_{j+1}\,-) = -\infty$. Then $\dom(\phi) \subset (-\infty,\tau_{j+1}]$, and we may define $\tilde{\phi}$ on $[\tau_j,\tau_{j+1}] \cap \R$ as described in part~(i) of Lemma~\ref{lem:linearize.phi.2}, where $a = \tau_j$ and $c = \tau_{j+1}$.

Suppose that $j < m$ and $\tau_{j+1} \in \dom(\phi)$ but either $\tau_j \not\in \dom(\phi)$ or $\phi'(\tau_j\,+) = \infty$. Then we may apply part~(i) Lemma~\ref{lem:linearize.phi.2} with $\phi(- \, \cdot)$ in place of $\phi$ and $a = -\tau_{j+1}$, $c = - \tau_j$.

If $\tau_j, \tau_{j+1} \in \dom(\phi)$ and both derivatives $\phi'(\tau_j\,+)$, $\phi'(\tau_{j+1}\,-)$ exist in $\R$, we may apply part~(ii) of Lemma~\ref{lem:linearize.phi.2} to define $\tilde{\phi}$ on $[\tau_j,\tau_{j+1}]$ such that it is piecewise linear with at most one change of slope in the interior while
\[
	\tilde{\phi}(\tau_j) = \phi(\tau_j) , \quad
	\tilde{\phi}'(\tau_j\,+) \le \phi'(\tau_j\,+) , \quad
	\tilde{\phi}(\tau_{j+1}) = \phi(\tau_{j+1}) \quad\text{and}\quad
	\tilde{\phi}'(\tau_{j+1}\,-) \ge \phi'(\tau_{j+1}\,-) .
\]

To complete the proof of property~(i), we have to modify $\tilde{\phi}$ in two cases: First suppose that $[\tau_{m-1},\infty) \subset\dom(\tilde{\phi})$. Then $\tilde{\phi}$ is linear on $[\tau_m,\infty)$, but it may have one change of slope within $(\tau_{m-1},\tau_m)$. If yes, we may redefine it on $[\tau_{m-1},\tau_m]$ to be linear such that $\tilde{\phi}(\tau_{m-1})$ remains the same, $\tilde{\phi}'(\tau_{m-1}\,+)$ becomes smaller and $\tilde{\phi}(\tau_m)$ becomes larger. Thereafter we may decrease the slope of $\tilde{\phi}$ on $(\tau_m,\infty)$ such that the original value of $\int_{\tau_m}^\infty e^{\tilde{\phi}(t)} \, dt$ is restored.

Secondly, suppose that $(\tau_j,\tau_{j+1}) \subset \R \setminus \bigcup_{i=1}^n [L_i,R_i]$ and $\tau_j, \tau_{j+1} \in \dom(\tilde{\phi})$ for some $j \le m-2$. If $\tilde{\phi}$ is not linear on $[\tau_j,\tau_{j+1}]$, then we define
\begin{align*}
	s \
	&:= \ \frac{\tilde{\phi}(\tau_{j+1}) - \tilde{\phi}(\tau_j)}{\tau_{j+1} - \tau_j} , \\
	\delta \
	&:= \ \int_{\tau_j}^{\tau_{j+1}} e^{\tilde{\phi}(t)} \, dt
		- \int_{\tau_j}^{\tau_{j+1}} e^{\tilde{\phi}(\tau_j) + s (t - \tau_j)} \, dt
\end{align*}
and
\[
	\check{\phi}(t)
	\ := \ \log \Bigl( \frac{1 - q}{1 - q - \delta} \Bigr)
		+ \begin{cases}
			\tilde{\phi}(t)
				& \text{if} \ t \in (-\infty,\tau_j] \cup [\tau_{j+1},\infty) , \\
			\tilde{\phi}(\tau_j) + s (t - \tau_j)
				& \text{if} \ t \in [\tau_j,\tau_{j+1}] .
		\end{cases}
\]
Then $(\check{\phi},q) \in \Theta$, too, and $\check{\phi} = \tilde{\phi} + \log((1 - q)/(1 - q - \delta))$ on $\bigcup_{i=1}^n [L_i,R_i] \cap \R$. Hence $\ell(\check{\phi},q) > \ell(\tilde{\phi},q)$, so we may replace $\tilde{\phi}$ with $\check{\phi}$.

Now we modify $\tilde{\phi}$ further, if necessary, such that it satisfies property~(ii) as well. If $\tilde{\phi}$ is not linear on $[\tau_j,\tau_\ell] \cap \R$, we may redefine it on $[\tau_j,\tau_\ell]$ as described in part~(i) of Lemma~\ref{lem:linearize.phi.2}. Then the inequalities in part~(iii) of Lemma~\ref{lem:linearize.phi.2} and our assumptions on the $\tau_k$, $j < k < \ell$, imply that this modification yields a larger value of $\ell(\tilde{\phi},q)$. Similarly one may enforce property~(iii).

Finally, if $2 \le j \le m-1$ such that $[\tau_{j-1},\tau_{j+1}] \subset \dom(\tilde{\phi})$, we may redefine $\tilde{\phi}$ on $[\tau_{j-1},\tau_{j+1}]$ as described in Lemma~\ref{lem:linearize.phi.1}, where $(a,b,c) = (\tau_{j-1},\tau_j,\tau_{j+1})$, without decreasing $\ell(\tilde{\phi},q)$. This proves property~(iv).
\end{proof}


\begin{proof}[\bf Proof of Lemma~\ref{lem:DSS}]
In case of a probability measure $M$, Lemma~\ref{lem:DSS} is just a special case of Theorem~2.2 of \cite{Duembgen_etal_2011}. If $\gamma := M(\R) \ne 1$, then $\tilde{M} := \gamma^{-1} M$ defines a probability measure on $\R$, and $S(M) = \bigl\{ x \in \R : 0 < \tilde{M}((-\infty,x]) < 1 \bigr\}$. Moreover, for any function $\phi \in \Phi$ and $\tilde{\phi} := \phi - \log \gamma$,
\[
	\int \tilde{\phi} \, d\tilde{M} - \int e^{\tilde{\phi}(x)} \, dx
	\ = \ \gamma^{-1} \Bigl( \int \phi \, dM - \int e^{\phi(x)} \, dx \Bigr)
		- \log \gamma .
\]
Consequently, $\phi \in \Phi$ maximizes $\int \phi \, dM - \int e^{\phi(x)} \, dx$ over $\Phi$ if, and only if, $\tilde{\phi} = \phi - \log \gamma$ maximizes $\int \tilde{\phi} \, d\tilde{M} - \int e^{\tilde{\phi}(x)} \, dx$. But $\dom(\tilde{\phi}) = \dom(\phi)$, and in case of $\tilde{\phi}$ being optimal, $1 = \int e^{\tilde{\phi}(x)} \, dx = \gamma^{-1} \int e^{\phi(x)} \, dx$, so $\int e^{\phi(x)} \, dx = M(\R)$.
\end{proof}


\begin{proof}[\bf Proof of Corollary~\ref{cor:consistency}]
For fixed $\eps > 0$ and real numbers $x \le y$, monotonicity of $F$ and $\hat{F}_n$ implies that
\[
	|\hat{F}_n - F| \ \ge \ \eps \ \text{on} \ [x,y]
	\quad\text{whenever}\quad
	\hat{F}_n(x) \ge F(y) + \eps \ \text{or} \ 
	F(x) \ge \hat{F}_n(y) + \eps .
\]
On the other hand, for $1 \le i \le n$ and $1 \le j \le M_{ni}$,
\[
	\sum_{k=1}^{M_{ni}+1}
		\bigl| (\hat{P}_n - P)(\TT_{n,i,k}) \bigr|
	\ \ge \ \bigl| (\hat{F}_n - F)(T_{n,i,j}) \bigr| .
\]
Thus $\hat{F}_n(x) \ge F(y) + \eps$ or $F(x) \ge \hat{F}_n(y) + \eps$ implies that
\[
	\frac{1}{n} \sum_{i=1}^n \sum_{k=1}^{M_{ni}+1}
		\bigl| (\hat{P}_n - P)(\TT_{n,i,k}) \bigr|
	\ \ge \ H_n([x,y]) \eps .
\]
If $\liminf_{n \to \infty} H_n([x,y]) > 0$, the latter inequality occurs by Theorem~\ref{thm:consistency} with asymptotic probability zero, which proves part~(i).

Part~(ii) is a simple consequence of part~(i) and continuity of $F$. For fixed $x \in (a,b)$ and $\delta > 0$, we know from part~(i) and the assumptions in part~(ii) that $\hat{F}_n(x) \le F(x+\delta) + o_p(1)$ and $\hat{F}_n(x) \ge F(x - \delta) + o_p(1)$. Since $F(x \pm \delta) \to F(x)$ as $\delta \downarrow 0$, this shows that $\hat{F}_n(x) \to_p F(x)$.
\end{proof}

Theorem~\ref{thm:consistency2} is closely related to results of \cite{Schuhmacher_Huesler_Duembgen_2009a} in the context of log-concave probability densities on $\R^d$. For the reader's convenience a self-contained proof is given here. We start with some elementary inequalities:

\begin{Lemma}
\label{lem:intervals.points}
Let $x_0 < x_1 < x_2 < x_3$ be real numbers such that
\[
	c_j \ := \ \log \frac{P([x_j,x_{j+1}])}{x_{j+1}-x_j} \ \in \ \R
\]
for $j = 0,1,2$. Then
\[
	\min\{c_0,c_2\} \ \le \ \phi \ \le \ 2c_1 - \min\{c_0,c_2\}
	\quad\text{on} \ [x_1, x_2] .
\]
\end{Lemma}

\begin{proof}[\bf Proof of Lemma~\ref{lem:intervals.points}]
For $j = 0,2$ let $z_j$ be a maximizer of $\phi$ over $[x_j,x_{j+1}]$. By concavity, the function $\phi$ is bounded from below by $\min\{\phi(z_0),\phi(z_2)\} \ge \min\{c_0,c_2\}$ on $[z_0, z_2] \supset [x_1,x_2]$. On the other hand, note first that for real numbers $x' < x''$,
\[
	\frac{P([x',x''])}{x'' - x'}
	\ \ge \ \sqrt{f(x') f(x'')} ,
\]
see \eqref{eq:helper1b}. Thus for $x \in [x_1,x_2]$,
\begin{align*}
	c_1 \
	&= \ \log \frac{ P([x_1, x]) + P([x,x_2])}{x_2 - x_1} \\
	&\ge \ \log \Bigl( \frac{x - x_1}{x_2 - x_1} \sqrt{f(x_1)f(x)}
		+ \frac{x_2 - x}{x_2 - x_1} \sqrt{f(x)f(x_2)} \Bigr) \\
	&\ge \ \log \sqrt{ f(x) \min\{f(x_1),f(x_2)\} } \\
	&\ge \ \bigl( \min\{c_0,c_2\} + \phi(x) \bigr) / 2 ,
\end{align*}
whence $\phi(x) \le 2 c_1 - \min\{c_0,c_2\}$.
\end{proof}

\begin{proof}[\bf Proof of Theorem~\ref{thm:consistency2}]
We first prove the assertions about the density estimator $\hat{f}_n$. Let $a_o$ and $b_o$ denote the infimum and supremum of $\dom(\phi) \cap (a,b)$, respectively. For any $x \in (a_o,b_o)$ and $\eps > 0$ there exists a $\delta = \delta(x,\eps) > 0$ such that $[x \pm 2 \delta] \subset (a_o,b_o)$ and $\bigl| \phi(y') - \phi(y'') \bigr| \le \eps$ for all $y',y'' \in [x \pm 2\delta]$. Now we apply Lemma~\ref{lem:intervals.points} to $x_0 := x - 2\delta$, $x_1 := x - \delta$, $x_2 := x + \delta$ and $x_3 := x + 2\delta$. One can easily verify that
\[
	|c_j - c_k| \ \le \ \eps
	\quad\text{and}\quad
	|\phi(y) - c_j| \ \le \ \eps
\]
for $j,k = 0,1,2$ and $y \in [x \pm 2\delta]$. Moreover, defining $\hat{c}_{nj}$ as $c_j$ with $\hat{P}_n$ in place of $P$, our assumption on $\hat{F}_n$ implies that with asymptotic probability $1$,
\[
	|\hat{c}_{nj} - c_j| \ \le \ \eps
	\quad\text{for} \ j = 0,1,2 .
\]
In this case, for any $y \in [x \pm \delta]$,
\begin{align*}
	\hat{\phi}_n(y) - \phi(y) \
	&\ge \ \min\{\hat{c}_{n0}, \hat{c}_{n2}\} - \phi(y)
		\ \ge \ \min\{c_0, c_2\} - \phi(y) - \eps
		\ \ge \ - 2 \eps , \\
	\hat{\phi}_n(y) - \phi(y) \
	&\le \ 2 \hat{c}_{n1} - \min\{\hat{c}_{n0}, \hat{c}_{n2}\} - \phi(y)
		\ \le \ 2 c_1 - \min\{c_0, c_2\} - \phi(y) + 3\eps
		\ \le \ 5 \eps .
\end{align*}
Consequently, for any $x \in (a_o,b_o)$ and $\eps > 0$ there exists a $\delta = \delta(x,\eps) > 0$ such that
\[
	\sup_{y \in [x \pm \delta]} \bigl| \hat{\phi}_n(y) - \phi(y) \bigr|
	\ \le \ 5 \eps
\]
with asymptotic probability $1$. These considerations prove that
\begin{equation}
\label{eq:fnhat.0}
	\sup_{x \in K_o} \bigl| \hat{f}_n(x) - f(x) \bigr|
	\ \to_p \ 0
	\quad\text{for any compact} \ K_o \subset (a_o,b_o) .
\end{equation}

Now we fix some point $x_o \in (a_o,b_o)$ and analyze $|\hat{f}_n - f|$ on $[x_o,b] \cap \R$. To this end we distinguish two different cases:

\textbf{Case~1: $b_o = b = \infty$ or $f(b_o) = 0$.}
We fix a point $b_* \in (x_o, b_o)$ such that $\phi(x_o) > \phi(b_*)$. It follows from \eqref{eq:fnhat.0} that $|\hat{f}_n - f| \to_p 0$ uniformly on $[x_o,b_*]$. Whenever $\hat{f}_n(b_*), \hat{f}_n(x_o) > 0$, it follows from concavity of $\phi$ and $\hat{\phi}_n$ that for $x \ge b_*$,
\[
	\max\{f(x),\hat{f}_n(x)\} \ \le \ \max\{f(b_*),\hat{f}_n(b_*)\}
		\exp \bigl( - \hat{\beta}_n(b_*) (x - b_*) \bigr)
\]
with
\begin{align*}
	\hat{\beta}_n(b_*) \
	&:= \ \frac{\min\{\phi(x_o),\hat{\phi}_n(x_o)\}
			- \max\{\phi(b_*),\hat{\phi}_n(b_*)\}}
		{b_* - x_o} \\
	&\to_p \ \frac{\phi(x_o) - \phi(b_*)}{b_* - x_o}
		\ =: \ \beta(b_*) \ > \ 0 .
\end{align*}
Consequently,
\begin{align*}
	\sup_{x \ge x_o} \bigl| \hat{f}_n(x) - f(x) \bigr| \
	&\le \ \sup_{x \ge b_*} \bigl| \hat{f}_n(x) - f(x) \bigr| + o_p(1)
		\ \le \ f(b_*) + o_p(1) , \\
	\int_{x_o}^\infty \bigl| \hat{f}_n(x) - f(x) \bigr| \, dx \
	&\le \ \int_{b_*}^\infty \bigl| \hat{f}_n(y) - f(y) \bigr| \, dy + o_p(1)
		\ \le \ f(b_*) / \beta(b_*) + o_p(1) .
\end{align*}
Since $\beta(b_*)$ is non-decreasing in $b_* > x_o$ and $\lim_{b_* \to \infty} f(b_*) = 0$, this shows that
\[
	\sup_{x \ge x_o} \bigl| \hat{f}_n(x) - f(x) \bigr|
	\ \to_p \ 0
	\quad\text{and}\quad
	\int_{x_o}^\infty \bigl| \hat{f}_n(x) - f(x) \bigr| \, dx
	\ \to_p \ 0 .
\]

\textbf{Case~2: $b_o < \infty$ and $f(b_o) > 0$.}
Here we fix an arbitrary point $b_* \in (x_o,b_o)$. Again, $|\hat{f}_n - f| \to_p 0$ uniformly on $[x_o,b_*]$. Moreover, by concavity of $\hat{\phi}_n$ and $\phi$, for any $x \in [b_*,b_o]$,
\begin{align*}
	\hat{f}_n(x) \
	&\le \ \hat{f}_n(b_*)
		\exp \Bigl( \bigl( \hat{\phi}_n(b_*) - \hat{\phi}_n(x_o) \bigr)
			(x - b_*)/(b_* - x_o) \Bigr) \\
	&\le \ \hat{f}_n(b_*)
		\exp \Bigl( \bigl( \hat{\phi}_n(b_*) - \hat{\phi}_n(x_o) \bigr)^+
			(b_o - b_*)/(b_* - x_o) \Bigr) \\
	&= \ \hat{f}_n(b_*)
		\max \bigl\{ \hat{f}_n(b_*)/\hat{f}_n(x_o), 1 \bigr\}^{(b_o-b_*)/(b_*-x_o)} \\
	&\to_p \ \overline{h}(b_*,b_o) \ := \ f(b_*)
		\max \bigl\{ f(b_*)/f(x_o), 1 \bigr\}^{(b_o-b_*)/(b_*-x_o)} ,
\end{align*}
and
\[
	\underline{h}(b_*,b_o) \ := \ \min \bigl\{ f(b_*), f(b_o) \bigr\}
	\ \le \ f(x)
	\ \le \ \overline{h}(b_*,b_o) .
\]
Thus
\begin{align*}
	\max_{x \in [x_o,b_o]} \bigl( \hat{f}_n(x) - f(x) \bigr)^+ \
	&= \ \max_{x \in [b_*,b_o]} \bigl( \hat{f}_n(x) - f(x) \bigr)^+ + o_p(1) \\
	&\le \ \overline{h}(b_*,b_o) - \underline{h}(b_*,b_o) + o_p(1) , \\
	\int_{x_o}^{b_o} \bigl| \hat{f}_n(x) - f(x) \bigr| \, dx \
	&\le \ \int_{b_*}^{b_o} \bigl| \hat{f}_n(x) - f(x) \bigr| \, dx + o_p(1) \\
	&\le \ (b_o - b_*) \overline{h}(b_*,b_o) + o_p(1) .
\end{align*}
Since $\underline{h}(b_*,b_o), \overline{h}(b_*,b_o) \to f(b_o)$ as $b_* \uparrow b_o$, these considerations show that for any fixed $\delta > 0$,
\[
	\max_{x \in [x_o,b_o]} \bigl( \hat{f}_n(x) - f(x) \bigr)^+
	\ \to_p \ 0 ,
	\quad
	\max_{x \in [x_o,b_o-\delta]} \bigl| \hat{f}_n(x) - f(x) \bigr|
	\ \to_p \ 0
\]
and
\[
	\int_{x_o}^{b_o} \bigl| \hat{f}_n(x) - f(x) \bigr| \, dx
	\ \to_p \ 0 .
\]
Suppose that in addition $b_o < b$, so $f \equiv 0$ on $(b_o,\infty)$. Then for fixed $\delta > 0$,
\[
	\sup_{x \in (b_o,b_o+\delta]} \hat{f}_n(y)
	\ \le \ \overline{h}(b_o,b_o + \delta) + o_p(1)
\]
with $\overline{h}(b_o,b_o + \delta) = f(b_o) \max \bigl\{ f(b_o)/f(x_o), 1 \bigr\}^{\delta/(b_o-x_o)}$, and
\[
	\hat{\pi}_n := \frac{\hat{P}_n([b_o,b_o + \delta])}{\delta}
	\ \to_p \ 0 .
\]
In particular, $\hat{\pi}_n < \hat{f}_n(x_o)$ with asymptotic probability one, and in that case we may conclude from concavity of $\hat{\phi}_n$ that $\hat{f}_n(b_o + \delta) \le \hat{\pi}_n$ and $\hat{f}_n(b_o + \delta + s) \le \hat{\pi}_n \bigl( \hat{\pi}_n/\hat{f}_n(x_o) \bigr)^{s/(b_o + \delta - x_o)}$ for any $s > 0$, so
\[
	\sup_{x \ge b_o + \delta} \hat{f}_n(x) \ \to_p \ 0
	\quad\text{and}\quad
	\int_{b_o + \delta}^\infty \hat{f}_n(x) \, dx \ \to_p \ 0 .
\]
Moreover,
\[
	\int_{b_o}^{b_o+\delta} \hat{f}_n(x) \, dx
	\ \le \ \delta \overline{h}(b_o,b_o + \delta) + o_p(1) .
\]
Since $\overline{h}(b_o,b_o + \delta) \to f(b_o)$ as $\delta \downarrow 0$, we may conclude that even
\[
	\int_{b_o}^\infty \hat{f}_n(x) \, dx
	\ \to_p \ 0 .
\]

Analogous arguments apply to $|\hat{f}_n - f|$ on the interval $[a,x_o] \cap \R$, and this yields the three claims about $\hat{f}_n - f$. Since
\[
	\bigl| \hat{F}_n(x) - F(x) \bigr|
	\ \le \ \bigl| \hat{F}_n(x_o) - F(x_o) \bigr|
		+ \int_a^b \bigl| \hat{f}_n(t) - f(t) \bigr| \, dt
\]
for any fixed $x_o \in (a_o,b_o)$ and arbitrary $x \in [a,b] \cap \R$, we also see that the supremum of $|\hat{F}_n - F|$ over $[a,b]$ converges to zero in probability.

It remains to prove the additional claims about $\hat{q}_n$. If $b = \infty$, then $\hat{q}_n = 1 - \hat{F}_n(\infty\,-) = 1 - F(\infty\,-) + o_p(1) = q + o_p(1)$. In case of $b < \infty$ we know that $\hat{F}_n(b) = F(b) + o_p(1)$, so $\hat{q}_n \le 1 - \hat{F}_n(b) = 1 - F(b) + o_p(1)$. We also know that $\hat{f}_n(b) \le f(b) + o_p(1)$ and $\hat{f}_n(x_o) \to_p f(x_o) > 0$ for any fixed $x_o \in (a_o,b_o)$. In case of $\hat{f}_n(x_o) > 0$ it follows from concavity of $\hat{\phi}_n = \log \hat{f}_n$ that $\hat{f}_n(b + s) \le \hat{f}_n(b) \bigl( \hat{f}_n(b)/\hat{f}_n(x_o) \bigr)^{s/(b - x_o)}$ for $s > 0$, and
\[
	\hat{q}_n
	\ \ge \ 1 - \hat{F}_n(b) - \hat{\eps}_n(x_o)
	\ = \ 1 - F(b) + o_p(1) - \hat{\eps}_n(x_o)
\]
with
\[
	\hat{\eps}_n(x_o) \ := \ \hat{f}_n(b)
		\int_0^\infty
			\bigl( \hat{f}_n(b)/\hat{f}_n(x_o) \bigr)^{s/(b - x_o)} \, ds .
\]
In case of $f(b) = 0$, one can easily verify that $\hat{\eps}_n(x_o) \to_p 0$. In case of $f(b) > 0 > \phi'(b\,-)$, we may choose $x_o$ such that $f(x_o) > f(b)$. For any such $x_o$,
\[
	\hat{\eps}_n(x_o) \ \to_p \ f(b) \Big/ \frac{\phi(x_o) - \phi(b)}{b - x_o} ,
\]
and the right hand side converges to $f(b) / |\phi'(b\,-)|$ as $x_o \uparrow b$.
\end{proof}

\paragraph{Acknowledgement.}
Constructive comments of the associate editor and three referees are gratefully acknowledged.

\bibliography{LogConc}

\begin{thebibliography}{}

\bibitem[Dempster et~al., 1977]{Dempster_etal_1977}
Dempster, A.~P., Laird, N.~M., and Rubin, D.~B. (1977).
\newblock Maximum-likelihood from incomplete data via the {EM} algorithm (with
  discussion).
\newblock {\em J. Royal Statist. Soc. Ser. B}, 39(1):1--38.

\bibitem[D{\"u}mbgen et~al., 2004]{Duembgen_Freitag_Jongbloed_2004}
D{\"u}mbgen, L., Freitag, S., and Jongbloed, G. (2004).
\newblock Consistency of concave regression with an application to
  current-status data.
\newblock {\em Math. Meth. Statist.}, 13:69--81.

\bibitem[D{\"u}mbgen et~al., 2006]{Duembgen_Freitag_Jongbloed_2006}
D{\"u}mbgen, L., Freitag-Wolf, S., and Jongbloed, G. (2006).
\newblock Estimating a unimodal distribution from interval-censored data.
\newblock {\em J. Amer. Statist. Assoc.}, 101:1094--1106.

\bibitem[D{\"u}mbgen et~al., 2011a]{Duembgen_etal_2007}
D{\"u}mbgen, L., H{\"u}sler, A., and Rufibach, K. (2007, revised 2011a).
\newblock Active set and {EM} algorithms for log-concave densities based on
  complete and censored data.
\newblock Technical report~61, IMSV, University of Bern.

\bibitem[D{\"u}mbgen and Rufibach, 2009]{Duembgen_Rufibach_2009}
D{\"u}mbgen, L. and Rufibach, K. (2009).
\newblock Maximum likelihood estimation of a log-concave density and its
  distribution function: basic properties and uniform consistency.
\newblock {\em Bernoulli}, 15(1):40--68.

\bibitem[D{\"u}mbgen et~al., 2011b]{Duembgen_etal_2011}
D{\"u}mbgen, L., Samworth, R., and Schuhmacher, D. (2011b).
\newblock Approximation by log-concave distributions, with applications to
  regression.
\newblock {\em Ann. Statist.}, 39(2):702--730.

\bibitem[Edmunson et~al., 1979]{Edmunson_etal_1979}
Edmunson, J.~H., Fleming, T.~R., Decker, D.~G., Malkasian, G.~D., Jefferies,
  J.~A., Webb, M.~J., and Kvols, L.~K. (1979).
\newblock Different chemotherapeutic sensitivities and host factors affecting
  prognosis in advanced ovarian carcinoma vs. minimal residual disease.
\newblock {\em Cancer Treatment Reports}, 63:241--247.

\bibitem[Fay, 2013]{Fay_2013}
Fay, M.~P. (2013).
\newblock {\em interval: Weighted Logrank Tests and NPMLE for interval censored
  data}.
\newblock R package, available at
  http://cran.r-project.org/web/packages/interval/\,.

\bibitem[{R Core Team}, 2013]{R_2013}
{R Core Team} (2013).
\newblock {\em R: A Language and Environment for Statistical Computing}.
\newblock R Foundation for Statistical Computing, Vienna, Austria.
\newblock Available at http://www.r-project.org/\,.

\bibitem[Schuhmacher et~al., 2011]{Schuhmacher_Huesler_Duembgen_2009a}
Schuhmacher, D., H{\"u}sler, A., and D{\"u}mbgen, L. (2011).
\newblock Multivariate log-concave distributions as a nearly parametric model.
\newblock {\em Statist. Risk Modeling}, 28(3):277--295.

\bibitem[Schuhmacher et~al., 2013]{Schuhmacher_etal_2013}
Schuhmacher, D., Rufibach, K., and D{\"u}mbgen, L. (2013).
\newblock {\em logconcens: Maximum likelihood estimation of a log-concave
  density based on censored data}.
\newblock R package, available at
  http://cran.r-project.org/web/packages/logconcens/\,.

\bibitem[Silverman, 1982]{Silverman_1982}
Silverman, B.~W. (1982).
\newblock On the estimation of a probability density function by the maximum
  penalized likelihood method.
\newblock {\em Ann. Statist.}, 10(3):795--810.

\bibitem[Therneau, 2013]{Therneau_2013}
Therneau, T. (2013).
\newblock {\em survival: Survival Analysis}.
\newblock R package, available at
  http://cran.r-project.org/web/packages/survival/\,.

\bibitem[Turnbull, 1976]{Turnbull_1976}
Turnbull, B.~W. (1976).
\newblock The empirical distribution function with arbitrarily grouped,
  censored and truncated data.
\newblock {\em J. Roy. Statist. Soc. Ser. B}, 38(3):290--295.

\bibitem[Walther, 2009]{Walther_2009}
Walther, G. (2009).
\newblock Inference and modeling with log-concave distributions.
\newblock {\em Statist. Sci.}, 24(3):319--327.

\end{thebibliography}
\bibliographystyle{apalike}

\end{document}